\documentclass{jpp}

\usepackage{tabularx}
\usepackage{hyperref}
\usepackage{amsmath}
\usepackage{amsfonts}
\usepackage{amssymb}
\usepackage{color}
\usepackage{xspace}
\usepackage{xfrac}

\DeclareSymbolFont{CMMI}{OML}{ccm}{m}{it}
\DeclareMathSymbol{v}{\mathalpha}{CMMI}{"76}

\newcommand{\gke}{\texttt{Gkeyll}\xspace}
\newcommand{\pderiv}[2]{
\frac{\partial #1}{\partial #2}
}
\newcommand{\pderivInline}[2]{
\partial #1/\partial #2
}
\renewcommand{\v}[1]{\ensuremath{\mbox{\boldmath$ #1 $}}} 
\newcommand{\gv}[1]{\ensuremath{\mbox{\boldmath$ #1 $}}} 
\newcommand{\uv}[1]{\ensuremath{\mathbf{\hat{#1}}}} 

\makeatletter
\newcommand{\dashover}[2][\mathop]{#1{\mathpalette\df@over{{\dashfill}{#2}}}}
\newcommand{\fillover}[2][\mathop]{#1{\mathpalette\df@over{{\solidfill}{#2}}}}
\newcommand{\df@over}[2]{\df@@over#1#2}
\newcommand\df@@over[3]{%
  \vbox{
    \offinterlineskip
    \ialign{##\cr
      #2{#1}\cr
      \noalign{\kern1pt}
      $\m@th#1#3$\cr
    }
  }%
}
\newcommand{\dashfill}[1]{%
  \kern-.5pt
  \xleaders\hbox{\kern.5pt\vrule height.4pt width \dash@width{#1}\kern.5pt}\hfill
  \kern-.5pt
}
\newcommand{\dash@width}[1]{%
  \ifx#1\displaystyle
    2pt
  \else
    \ifx#1\textstyle
      1.5pt
    \else
      \ifx#1\scriptstyle
        1.25pt
      \else
        \ifx#1\scriptscriptstyle
          1pt
        \fi
      \fi
    \fi
  \fi
}
\newcommand{\solidfill}[1]{\leaders\hrule\hfill}
\makeatother

\newtheorem{proposition}{Proposition}

\journal{Journal of Plasma Physics}

\begin{document}

\title[Electromagnetic gyrokinetics in the tokamak edge]{Electromagnetic full-$f$ gyrokinetics in the tokamak edge with discontinuous Galerkin methods}
\author[N. R. Mandell, A. Hakim, G. W. Hammett and M. Francisquez]{N.~R. Mandell\aff{1} \corresp{\email{nmandell@princeton.edu}}, A. Hakim\aff{2}, G.~W. Hammett\aff{2}, M. Francisquez\aff{3}}
\affiliation{
\aff{1}Department of Astrophysical Sciences, Princeton University, Princeton, NJ 08543, USA
\aff{2}Princeton Plasma Physics Laboratory, Princeton, NJ 08543, USA
\aff{3}MIT Plasma Science and Fusion Center, Cambridge, MA, 02139, USA
}

\maketitle

\begin{abstract}
    We present an energy-conserving discontinuous Galerkin scheme for the full-$f$ electromagnetic gyrokinetic system in the long-wavelength limit. We use the symplectic formulation and solve directly for $\partial A_\parallel/\partial t$, the inductive component of the parallel electric field, using a generalized Ohm's law derived directly from the gyrokinetic equation. Linear benchmarks are performed to verify the implementation and show that the scheme avoids the Amp\`ere cancellation problem. We perform a nonlinear electromagnetic simulation in a helical open-field-line system as a rough model of the tokamak scrape-off layer using parameters from the National Spherical Torus Experiment (NSTX). This is the first published nonlinear electromagnetic gyrokinetic simulation on open field lines. Comparisons are made to a corresponding electrostatic simulation.
\end{abstract}

\section{Introduction}

Understanding turbulent transport physics in the tokamak edge and scrape-off layer (SOL) is critical to developing a successful fusion reactor. The dynamics in these regions plays a key role in determining the L--H transition, the pedestal height, and the heat load to the vessel walls. While the edge is often modelled by Braginskii-type fluid models that have provided valuable results and insights \citep{xu2008boundary,tamain2010tokam,ricci2012simulation,zhu2017global,francisquez2017global}, a kinetic treatment will inevitably be necessary for reliable quantitative predictions in some cases \citep{jenko2001nonlinear,cohen2008progress}. Gyrokinetic theory and direct numerical simulation have become important tools for studying turbulence and transport in fusion plasmas, especially in the core region \citep{parker1993gyrokinetic,kotschenreuther1995comparison,lin2000gyrokinetic,dimits2000comparisons,dorland2000electron,jenko2000massively,jost2001global,candy2003eulerian,idomura2003global,watanabe2005velocity,jolliet2007global,idomura2008conservative,peeters2009nonlinear,lanti2019orb5}. In the edge and SOL, gyrokinetic simulations are particularly challenging because the large, intermittent fluctuations in the SOL make assumptions of scale separation between equilibrium and fluctuations not strongly valid. This necessitates a full-$f$ approach that self-consistently evolves the full distribution function, $f$ (as opposed to the $\delta f$ approach commonly used in the core, where one assumes $f=F_0+\delta f$ with a fixed background $F_0$ so that only $\delta f$ perturbations must be evolved, and the parallel electric field nonlinearity is frequently neglected). Steady progress in gyrokinetic edge/SOL modelling has been made with both particle-in-cell (PIC) \citep{ku2009full,korpilo2016gyrokinetic,ku2016new} and continuum \citep{shi2019full,pan2018full,dorf2016continuum} methods. Another challenge is the magnetic geometry of the edge/SOL region, which requires treatment of open and closed magnetic field-line regions and the resulting plasma interactions with material walls on open field lines. The X-point in a diverted geometry is an additional complication which makes the use of field-aligned coordinates challenging. Currently only the XGC1 hybrid-Lagrangian PIC code \citep{ku2016new} can simulate gyrokinetic turbulence in a three-dimensional diverted geometry with an X-point.


The edge/SOL region also features steep pressure gradients, especially in the H-mode transport barrier and SOL regions, which contributes to the importance of electromagnetic effects. In this regime, the parallel electron dynamics is no longer fast relative to the drift turbulence, so electrons can no longer be treated adiabatically \citep{scott1997three}. 
This leads to coupling of the perpendicular vortex motions and kinetic shear Alfv\'en waves, which results in field-line bending \citep{xu2010intermittent}.
Including electromagnetic effects in gyrokinetic simulations has proved numerically and computationally challenging, both in the core and in the edge. The so-called Amp\`ere cancellation problem is one of the main numerical issues that has troubled primarily PIC codes
\citep{reynders1993gyrokinetic,cummings1994gyrokinetic}. Various $\delta f$ PIC schemes to address the cancellation problem have been developed and there are interesting recent advances in this area \citep{chen2003deltaf,mishchenko2004conventional,hatzky2007electromagnetic,mishchenko2014pullback,startsev2014finite,bao2018conservative}.
Meanwhile, some continuum $\delta f$ core codes avoided the cancellation problem completely \citep{rewoldt1987collisional,kotschenreuther1995comparison}, while others had to address somewhat minor issues resulting from it \citep{jenko2000massively,candy2003eulerian}. 
Meanwhile, some continuum $\delta f$ core codes avoided the cancellation problem completely \citep{rewoldt1987collisional,kotschenreuther1995comparison}, while others had to address somewhat minor issues resulting from it \citep{jenko2000massively,candy2003eulerian}. 
With respect to the cancellation problem, one possible reason for the differences might be that in continuum codes the fields and particles are discretized on the same grid, whereas in PIC codes the particle positions do not coincide with the field grid. 
Because particle positions are randomly located relative to the field grid, one might need to be more careful in some way when treating the interaction of the particles and electromagnetic fields.

To this point, all published nonlinear electromagnetic gyrokinetic results have focused on the core region, mostly within the $\delta f$ formulation neglecting the $E_\parallel$ nonlinearity, although the ORB5 PIC code includes the $E_\parallel$ nonlinearity and is effectively full-$f$ \citep{lanti2019orb5}. 
The XGC1 code is also full-$f$ and is focused on both the core and the edge/SOL; it has an option for a gyrokinetic ion/drift-fluid massless electron hybrid model \citep{hager2017verification}, with a fully kinetic implicit electromagnetic scheme based on \citet{chen2015multi} recently implemented and under further development \citep{ku2018fully}. 
Other gyrokinetic codes working on the SOL are not yet electromagnetic. Thus, to our knowledge, the results presented here are the first nonlinear electromagnetic full-$f$ gyrokinetic turbulence simulations on open field lines.

In this paper we present a numerical scheme for simulating the full-$f$ electromagnetic gyrokinetic system using a continuum approach. We use an energy-conserving discontinuous Galerkin (DG) scheme for the discretization of the gyrokinetic system in phase space, building on the work of \citet{hakim2019discontinuous,shi2017gyrokinetic,shi-thesis,liu2000high}. DG methods are attractive because they are highly local (enabling fairly straightforward parallelization schemes), allow high-order accuracy, and enforce local conservation laws \citep{durran2010numerical}. The present target of the scheme is simulating the edge and SOL of tokamaks, although the scheme could in principle be used for whole-device modelling including the core. Our scheme has been implemented as part of the gyrokinetics solver \citep{shi2017gyrokinetic,shi2019full,bernard2019gyrokinetic} of the \gke computational plasma framework, which also includes solvers for the Vlasov--Maxwell system \citep{cagas2017continuum,juno2018discontinuous} and multi-moment fluid equations \citep{wang2015comparison}.  

The paper is organized as follows. In Section \ref{sec:emgk}, we describe the electromagnetic gyrokinetic system and some of its conservation properties. Section \ref{sec:DG} describes the discontinuous Galerkin phase-space discretization of the system, and also presents proofs that the scheme preserves particle and energy conservation. The time-discretization scheme is handled in Section \ref{sec:timedisc}. In Section \ref{sec:linear} we present some linear electromagnetic benchmarks that validate the scheme and also demonstrate the avoidance of the cancellation problem. We present nonlinear results showing the first electromagnetic gyrokinetic turbulence simulation on open field lines in Section \ref{sec:nonlinear}, along with comparisons to a corresponding electrostatic simulation. We summarize and address future work in Section \ref{sec:conclusion}.

%
%
%
%
%
%
%
%

\section{The electromagnetic gyrokinetic system}
\label{sec:emgk}
\subsection{Basic equations}
We solve the full-$f$ electromagnetic gyrokinetic (EMGK) equation in the symplectic formulation \citep{brizard2007foundations}, which describes the evolution of the gyrocentre distribution function $f_s(\v{Z},t)=f_s(\v{R},v_\parallel,\mu,t)$ for each species $s$, where $\v{Z}$ is a phase-space coordinate composed of the guiding centre position $\v{R}=(x,y,z)$, the parallel velocity $v_\parallel$ and the magnetic moment $\mu=m_s v_\perp^2/(2B)$.
In terms of the gyrocentre Hamiltonian and the Poisson bracket in gyrocentre coordinates, {and also including collisions $C[f_s]$ and sources $S_s$ (which do not derive from the bracket),} the gyrokinetic equation is given by\footnote{One can use extended gyrocentre phase-space coordinates, which include time $t$ and the canonically conjugate energy $w$, to include the time derivative terms in Eq. (\ref{liouville}) inside an extended Poisson bracket \citep{brizard2007foundations}. For ease of presentation we do not take this approach.}
\begin{align}
    \pderiv{f_s}{t} + \{f_s, H_s\} - \frac{q_s}{m_s}\pderiv{A_\parallel}{t}\pderiv{f_s}{v_\parallel} = C[f_s] + S_s, \label{liouville}
\end{align}
or equivalently,
\begin{align}
    \pderiv{f_s}{t} + \dot{\v{R}}\v{\cdot}\nabla f_s + \dot{v}^H_\parallel \pderiv{f_s}{v_\parallel}- \frac{q_s}{m_s}\pderiv{A_\parallel}{t}\pderiv{f_s}{v_\parallel} = C[f_s] + S_s,
\end{align}
where the gyrokinetic Poisson bracket is given by
\begin{equation}
  \{F,G\} = \frac{\v{B^*}}{m B_\parallel^*} \v{\cdot} \left(\nabla F \frac{\partial G}{\partial v_\parallel} - \frac{\partial F}{\partial v_\parallel}\nabla G\right) - \frac{\uv{b}}{q B_\parallel^*}\times \nabla F \v{\cdot} \nabla G, 
\end{equation}
and we take the gyrocentre Hamiltonian to be
\begin{equation}
    H_s = \frac{1}{2}m_sv_\parallel^2 + \mu B + q_s  \phi.
\end{equation}
Here we have taken the long-wavelength (drift-kinetic) limit to neglect gyroaveraging of the electrostatic potential $\phi$, and we have also dropped higher- order terms in the Hamiltonian that appear in \emph{e.g.} \citet{brizard2007foundations}; extensions to include gyroaveraging will be included in later work, but these additions will not change the overall scheme presented here.
The nonlinear phase-space characteristics are given by
\begin{gather}
    \dot{\v{R}} = \{\v{R},H_s\} = \frac{\v{B^*}}{B_\parallel^*}v_\parallel + \frac{\uv{b}}{q_s B_\parallel^*}\times\left(\mu\nabla B + q_s \nabla \phi\right), \\
    \dot{v}_\parallel = \dot{v}^H_\parallel -\frac{q_s}{m_s}\pderiv{A_\parallel}{t} = \{v_\parallel,H_s\}-\frac{q_s}{m_s}\pderiv{A_\parallel}{t} = -\frac{\v{B^*}}{m_s B_\parallel^*}\v{\cdot}\left(\mu\nabla B + q_s \nabla \phi\right)-\frac{q_s}{m_s}\pderiv{A_\parallel}{t}. \label{vpardot}
\end{gather}
Here, $B_\parallel^*=\uv{b}\v{\cdot} \v{B^*}$ is the parallel component of the effective magnetic field $\v{B^*}=\v{B}+(m_s v_\parallel/q_s)\nabla\times\uv{b} + \delta \v{B}$, where $\v{B} = B \uv{b}$ is the equilibrium magnetic field and $\delta \v{B} = \nabla\times(A_\parallel \uv{b})\approx \nabla A_\parallel \times \uv{b}$ is the perturbed magnetic field {(assuming that the equilibrium magnetic field varies on spatial scales longer than perturbations so that $A_\parallel\nabla\times\uv{b}$ can be neglected)}. We neglect higher-order parallel compressional fluctuations of the magnetic field, so that $\delta\v{B}=\delta\v{B}_\perp$. The species charge and mass are $q_s$ and $m_s$, respectively. In Eq. (\ref{vpardot}), note that we have separated $\dot{v}_\parallel$ into a term that comes from the Hamiltonian, $\dot{v}^H_\parallel = \{v_\parallel,H_s\}$, and another term proportional to the inductive component of the parallel electric field, $(q/m)\partial {A_\parallel}/\partial{t}$. We use this notation for convenience, and so that the time derivative of the parallel vector potential $A_\parallel$ appears explicitly. {Further, we will assume a field-aligned coordinate system \citep[\emph{e.g.}][]{beer1995field}, and we will take the perpendicular directions to be $x$ and $y$, and the parallel direction to be $z$.}

In the absence of collisions $C[f_s]$ and sources $S_s$, Eq. \eqref{liouville} can be recognized as a Liouville equation, which shows that the distribution function is conserved along the nonlinear characteristics. Liouville's theorem also shows that phase-space volume is conserved,
\begin{equation}
    \pderiv{\mathcal{J}}{t} + \nabla\v{\cdot}\left(\mathcal{J}\dot{\v{R}}\right) + \pderiv{}{v_\parallel}\left(\mathcal{J}\dot{v}^H_\parallel\right) - \pderiv{}{v_\parallel}\left(\mathcal{J}\frac{q_s}{m_s}\pderiv{A_\parallel}{t}\right) = 0,
\end{equation}
where $\mathcal{J} = B_\parallel^*$ is the Jacobian of the gyrocentre coordinates, and we will make the approximation $\uv{b}\v{\cdot}\nabla\times\uv{b}\approx0$ so that $B_\parallel^*\approx B$. 

We can now write the gyrokinetic equation in conservative form, 
\begin{equation}
\pderiv{(\mathcal{J}f_s)}{t} + \nabla\v{\cdot}( \mathcal{J} \dot{\v{R}} f_s) + \pderiv{}{v_\parallel}\left(\mathcal{J}\dot{v}^H_\parallel f_s\right)- \pderiv{}{v_\parallel}\left(\mathcal{J}\frac{q_s}{m_s}\pderiv{A_\parallel}{t} f_s \right)= \mathcal{J} C[f_s] + \mathcal{J} S_s. \label{emgk} 
\end{equation}
Here, we have used the symplectic formulation of electromagnetic gyrokinetics, where the parallel velocity is used as an independent variable (as opposed to the Hamiltonian formulation which uses the parallel canonical momentum $p_\parallel$ as an independent variable) \citep{brizard2007foundations,hahm1988nonlinear}. Notably, in the symplectic formulation, the time derivative of $A_\parallel$ appears explicitly in the gyrokinetic equation, Eq. \eqref{emgk}, and $A_\parallel$ appears in $\v{B^*}$ but not in the Hamiltonian.  

The electrostatic potential is determined by the quasi-neutrality condition in the long-wavelength limit, given by
\begin{equation}
    \sigma_g + \sigma_\text{pol} = \sigma_g - \nabla\v{\cdot}\v{P} = 0,
\end{equation}
with the guiding centre charge density (neglecting gyroaveraging in the long-wavelength limit)
\begin{equation}
    \sigma_g = \sum_s q_s \int d\v{w}\ \mathcal{J} f_s.
\end{equation}
Here, we have defined $d\v{w}= 2\pi\, m_s^{-1} dv_\parallel\, d\mu = m_s^{-1} dv_\parallel\, d\mu \int d\alpha$ as the gyrocentre velocity-space volume element
$(d{\bf v}=m_s^{-1}dv_\parallel\, d\mu\, d\alpha\, \mathcal{J})$ with the gyroangle $\alpha$ integrated away and the Jacobian factored out. The polarization vector is
\begin{equation}
    \v{P} = -\sum_s \int d\v{w}\ \frac{m_s}{B^2}\mathcal{J}f_s \nabla_\perp \phi \approx -\sum_s \frac{m_s n_{0s}}{B^2} \nabla_\perp \phi \equiv - \epsilon_\perp \nabla_\perp \phi,
\end{equation}
where $\nabla_\perp=\nabla-\uv{b}(\uv{b}\v{\cdot}\nabla)$ is the gradient perpendicular to the background magnetic field. We use a linearized polarization density $n_0$ that we take to be a constant in time, which is consistent with neglecting a second-order $E\times B$ energy term in the Hamiltonian. While the validity of this approximation in the SOL can be questioned due to large density fluctuations, a linearized polarization density is commonly used for computational efficiency \citep{ku2018fast,shi2019full}. 
Future work will include the nonlinear polarization density along with the second-order $E\times B$ energy term in the Hamiltonian.
The quasi-neutrality condition can then be rewritten as the long-wavelength gyrokinetic Poisson equation,
\begin{equation}
    -\nabla \v{\cdot} \sum_s \frac{m_s n_{0s}}{B^2} \nabla_\perp \phi = \sum_s q_s \int d\v{w}\ \mathcal{J}f_s. \label{poisson}
\end{equation}
Even in the long-wavelength limit with no gyroaveraging, the first-order polarization charge density on the left-hand side of Eq. (\ref{poisson}) incorporates some finite Larmor radius (FLR) effects.

The parallel vector potential $A_\parallel$ is determined by the parallel Amp\`ere equation,
\begin{equation}
    -\nabla^2_\perp A_\parallel = \mu_0 J_\parallel = \mu_0 \sum_s q_s \int d\v{w}\ \mathcal{J} v_\parallel f_s. \label{ampere}
\end{equation}
Note that we can also take the time derivative of this equation to get a generalized Ohm's law which can be solved directly for $\pderivInline{A_\parallel}{t}$, the inductive component of the parallel electric field $E_\parallel$ \citep{reynders1993gyrokinetic, cummings1994gyrokinetic, chen2001gyrokinetic}
\begin{equation}
    -\nabla_\perp^2 \pderiv{A_\parallel}{t} = \mu_0 \sum_s q_s \int d\v{w}\ v_\parallel \pderiv{(\mathcal{J} f_s)}{t}.
\end{equation}
Writing the gyrokinetic equation as
\begin{equation}
    \pderiv{(\mathcal{J}f_s)}{t} = 
    \pderiv{(\mathcal{J}f_s)}{t}^\star + \pderiv{}{v_\parallel}\left(\mathcal{J} \frac{q_s}{m_s}\pderiv{A_\parallel}{t} f_s\right), \label{fstar}
\end{equation}
where $\partial{(\mathcal{J}f_s)^\star}/\partial{t}$ denotes all the terms in the gyrokinetic equation (including sources and collisions) except the $\pderivInline{A_\parallel}{t}$ term, Ohm's law can be rewritten (after an integration by parts) as 
\begin{equation}
    \left(-\nabla_\perp^2 + \sum_s \frac{\mu_0 q_s^2}{m_s} \int d\v{w}\  \mathcal{J} f_s\right) \pderiv{A_\parallel}{t} = \mu_0 \sum_s q_s \int d\v{w}\ v_\parallel \pderiv{(\mathcal{J}f_s)}{t}^\star. \label{ohmstar}
\end{equation}
As we will show in Section \ref{sec:timedisc}, this form allows for the use of an explicit time-stepping scheme in which one can first compute  $\partial{(\mathcal{J}f_s)^\star}/\partial{t}$ (which does not involve $\pderivInline{A_\parallel}{t})$, then compute $\pderivInline{A_\parallel}{t}$, and finally compute $\pderivInline{(\mathcal{J}f_s)}{t}$.
Note, however, that in some PIC approaches \citep{reynders1993gyrokinetic, chen2001gyrokinetic}, one must expand the right-hand side of Eq. (\ref{ohmstar}) by inserting the gyrokinetic equation so that the right-hand side involves only moments of $f_s$ without time derivatives. In our continuum scheme we can compute  $\partial{(\mathcal{J}f_s)^\star}/\partial{t}$ directly and then perform the integration. Further, note that, although we are using the symplectic ($v_\parallel$) formulation of EMGK, our Ohm's law from Eq. (\ref{ohmstar}) contains two integral terms which must cancel exactly. This is the root of the cancellation problem that appears in Amp\`ere's law in the Hamiltonian ($p_\parallel$) formulation, and in Appendix \ref{app:cancel} we show that the same cancellation problem could arise from Eq. (\ref{ohmstar}) if the integrals are not treated consistently.

To model the effect of collisions we use a conservative Lenard--Bernstein (or Dougherty) collision operator \citep{Lenard1958plasma,Dougherty1964model},
\begin{align} \label{eq:GkLBOEq}
\mathcal{J}C[f] &= \nu\left\lbrace\pderiv{}{v_\parallel}\left[\left(v_\parallel - u_\parallel\right)\mathcal{J} f+v_{t}^2\pderiv{(\mathcal{J} f)}{v_\parallel}\right]+\pderiv{}{\mu}\left[2\mu \mathcal{J} f+2\mu\frac{m}{B}v_{t}^2\pderiv{(\mathcal{J} f)}{\mu}\right]\right\rbrace,
\end{align}
where
\begin{align}
    n u_\parallel = \int d\v{w}\ \mathcal{J} v_\parallel f, \qquad\qquad n u_\parallel^2 + 3 n v_{t}^2 = \int d\v{w}\ \mathcal{J}\left(v_\parallel^2 + 2\mu B/m\right)f,
\end{align}
with $n = \int d\v{w}\ \mathcal{J} f$.
This collision operator contains the effect of drag and pitch-angle scattering, and it conserves number, momentum and energy density. Consistent with our present long-wavelength treatment of the gyrokinetic system, finite Larmor radius effects are ignored. For simplicity we restrict ourselves to the case in which the collision frequency $\nu$ is velocity independent, i.e. $\nu\neq\nu(v)$.  
Further details about this collision operator, including its conservation properties and its discretization, are left to a separate paper \citep{francisquez2020conservative}. In this work, we include only the effects of like-species collisions, which neglects electron--ion collisions and the resulting resistivity. A conservative scheme for cross-species collisions has also been implemented and will be included in later work. Extensions to a more complete collision operator are in progress.

\subsection{Conservation properties}

In the absence of collisions and sources, the Hamiltonian structure of the gyrokinetic system guarantees conservation of arbitrary functions of $f$ along the characteristics,
\begin{align}
    \pderiv{G(f)}{t} + \{G(f), H\} - \frac{q}{m}\pderiv{A_\parallel}{t}\pderiv{G(f)}{v_\parallel}=0,
\end{align}
along with corresponding Casimir invariants $\int d\v{R}\ d\v{w}\ \mathcal{J} G(f)$, where $d\v{R}= dx\,dy\,dz$.
Thus, the system has an infinite number of conserved quantities such as the total particle number (or $L_1$ norm) $N=\int d\v{R}\ d\v{w}\ \mathcal{J} f$, the $L_2$ norm $M=\int d\v{R}\ d\v{w}\ \mathcal{J} f^2$ and the kinetic entropy $S=-\int d\v{R}\ d\v{w}\ \mathcal{J} f\ln f$ \citep{idomura2008conservative}.

The system also conserves total energy, $W = W_K + W_E + W_B = W_H - W_E + W_B$, where the kinetic particle energy (neglecting the kinetic energy of the $E\times B$ flow) is
\begin{equation}
    W_K = \sum_s \int d\v{R}\ d\v{w}\ \mathcal{J} \left(\frac{1}{2}m_s v_\parallel^2 + \mu B\right) f_s, \label{W_k}
\end{equation}
the (non-vacuum) electrostatic field energy (equivalent to the kinetic energy associated with the $E\times B$ flow of particles) is
\begin{equation}
    W_E = \sum_s \int d\v{R}\ \frac{1}{2}\frac{m_s n_{0s}}{B^2}|\nabla_\perp\phi|^2 = \int d\v{R}\ \frac{\epsilon_\perp}{2}|\nabla_\perp\phi|^2 , \label{W_E}
\end{equation}
the (perturbed) electromagnetic field energy is
\begin{equation}
    W_B = \int d\v{R}\ \frac{1}{2\mu_0}|\nabla_\perp A_\parallel|^2, \label{W_B}
\end{equation}
and
\begin{equation}
    W_H = \sum_s \int d\v{R}\ d\v{w}\ \mathcal{J} H_s f_s. \label{W_H}
\end{equation}
(Note that $W_H$ is the sum of the particle kinetic energy and twice the potential energy, because every pair of particle interactions is double counted in the raw integral of $q_s \phi f_s$.)

Assuming the boundary conditions are periodic or that the distribution function vanishes at the boundary so that surface terms vanish,
the evolution of these quantities can be calculated as
\begin{align}
    \frac{dW_H}{dt} &= \sum_s \int d\v{R}\ d\v{w}\  H_s \pderiv{(\mathcal{J}f_s)}{t} + 
    \sum_s \int d\v{R}\ d\v{w}\   \mathcal{J}f_s \pderiv{H_s}{t} \notag\\&= - \int d\v{R}\  J_\parallel \pderiv{A_\parallel}{t} + \int d\v{R}\ \sigma_g \pderiv{\phi}{t}, \label{dWHdt} \\
   \frac{d W_E}{dt} &= \sum_s \int d\v{R}\ \frac{m_s n_{0s}}{B^2}\nabla_\perp\phi \v{\cdot} \nabla_\perp \pderiv{\phi}{t} = \int d\v{R}\ \sigma_g \pderiv{\phi}{t}, \label{dWEdt} \\
   \frac{d W_B}{d t} &= \int d\v{R}\ \frac{1}{\mu_0}\nabla_\perp A_\parallel \v{\cdot} \nabla_\perp \pderiv{A_\parallel}{t} = \int d\v{R}\ J_\parallel \pderiv{A_\parallel}{t}, \label{dWBdt}
\end{align}
so that total energy is indeed conserved:
\begin{align}
    \frac{dW}{dt} = \frac{dW_H}{dt} - \frac{d W_E}{dt} + \frac{d W_B}{dt} = 0.
\end{align}




\section{The discrete EMGK system}
\label{sec:DG}
\noindent
In this section we describe the  phase-space discretization of the electromagnetic gyrokinetic system used in \gke. 

\subsection{Discrete equations}

We use an energy-conserving discontinuous Galerkin scheme to discretize the gyrokinetic system in phase space. The scheme generalizes the algorithm of \citet{liu2000high} (originally for the two-dimensional incompressible Euler and Navier-Stokes equations) to arbitrary Hamiltonian systems \citep{hakim2019discontinuous,shi2017gyrokinetic,shi-thesis}. However, unlike the nodal approach used in \citet{shi2017gyrokinetic,shi-thesis}, we use a modal DG scheme.

We start by decomposing the global phase-space domain $\Omega$ into a {structured} phase-space mesh $\mathcal{T}$ with cells $\mathcal{K}_j \in \mathcal{T},\ j=1,...,N$. We then introduce a piecewise-polynomial approximation space for the distribution function $f(\v{R},v_\parallel,\mu)$,
\begin{equation}
    \mathcal{V}_h^p = \{v : v|_{\mathcal{K}_j}\in \v{P}^p, \forall \mathcal{K}_j\in \mathcal{T}\},
\end{equation}
where $\v{P}^p$ is some space of polynomials with maximum degree $p$ (by some measure). 
That is, $v(z)$ are polynomial functions of $\v{z}$ in each cell, and $\v{P}^p$ is the space of the linear combination of some set of multi-variate polynomials.
In this work, we choose $\v{P}^p$ to be an orthonormalized serendipity polynomial element space \citep{arnold2011serendipity}.
The serendipity basis set
has the advantage of using fewer basis functions while giving the same formal
convergence order (although it is less accurate) as the Lagrange tensor
basis, although note that, for $p=1$, the serendipity basis is equivalent to the Lagrange tensor basis. 
We can then obtain the discrete weak form of the gyrokinetic equation by multiplying Eq. (\ref{emgk}) by any test function $\psi\in\mathcal{V}^p_h$
and integrating (by parts) in each cell
\begin{align}
    \int\limits_{\mathcal{K}_j}d\v{R}\ d\v{w}\ &\psi \pderiv{(\mathcal{J}f_h)}{t} + \oint\limits_{\partial \mathcal{K}_j} d\v{w}\ d\v{s}_R\v{\cdot}\dot{\v{R}}_h \psi^- \widehat{\mathcal{J}f_h}+ \oint\limits_{\partial \mathcal{K}_j} d\v{R}\ ds_w\ \left(\dot{v}^H_{\parallel h}-\frac{q}{m}\pderiv{A_{\parallel h}}{t}\right) \psi^- \widehat{\mathcal{J}f_h} \notag\qquad\qquad\\
    &\quad- \int\limits_{\mathcal{K}_j}d\v{R}\ d\v{w}\ \mathcal{J}f_h \dot{\v{R}}_h\v{\cdot} \nabla \psi - \int\limits_{\mathcal{K}_j}d\v{R}\ d\v{w}\ \mathcal{J}f_h \left(\dot{v}^H_{\parallel h}-\frac{q}{m}\pderiv{A_{\parallel h}}{t}\right) \pderiv{\psi}{v_\parallel} \notag\\&\qquad= \int\limits_{\mathcal{K}_j}d\v{R}\ d\v{w}\ \psi\left(\mathcal{J}C[f_h] + \mathcal{J}S\right). \label{DGgk}
\end{align}
Solving this equation for all test functions $\psi\in\mathcal{V}_h^p$  in all cells $\mathcal{K}_j\in\mathcal{T}$ yields the discretized distribution function $f_h\in\mathcal{V}^p_h$, where the subscript $h$ denotes a discrete quantity in $\mathcal{V}^p_h$. In the surface terms, $d\v{s}_R$ is the differential element on a configuration-space surface (pointing outward normal to the surface), $ds_w=2\pi\, m_s^{-1} d\mu\, \v{n}\v{\cdot} (\partial{\v{Z}}/\partial{v_\parallel})$ is the differential element on a $v_\parallel$ surface, and the notation $\psi^-$ ($\psi^+$) indicates that the function $\psi$ is evaluated just inside (outside) the surface $\partial \mathcal{K}_j$. The notation $\widehat{f}=\widehat{f}(f^+,f^-)$ indicates a `numerical flux', which takes a single value at the cell surface and in general can depend on the solution on both sides of the surface since the solution is discontinuous at the surface. Here, we choose to use standard upwind fluxes, which depend on the local value of the phase-space characteristic flow normal to the surface evaluated at each Gaussian quadrature point on the surface. Denoting the flow as $\gv{\alpha}_h$, the upwind flux can be expressed as
\begin{equation}
    \widehat{f_h} = \frac{1}{2}\left(f_h^++f_h^-\right)-\frac{1}{2}\text{sgn}\left(\v{n}\v{\cdot}\gv{\alpha}_h\right)\left(f_h^+-f_h^-\right),
\end{equation}
where $\v{n}=d\v{s}/|d\v{s}|$ is the unit normal pointing out of the $\partial \mathcal{K}_j$ surface. 

We will introduce a subset of $\mathcal{V}^p_h$ where the piecewise polynomials are continuous across cell interfaces, denoted by $\fillover{\mathcal{V}}^p_h$.
As we will show later, in order to preserve energy conservation in our discrete scheme, we will require that the discrete Hamiltonian be continuous across cell interfaces, \textit{i.e.} $H_h\in\fillover{\mathcal{V}}^p_h$  \citep{hakim2019discontinuous,liu2000high,shi2017gyrokinetic,shi-thesis}. Note that one can show that this ensures that the discrete phase-space characteristics, $\dot{\v{R}}_h=\{\v{R},H_h\}$ and $\dot{v}^H_{\parallel h}-({q_s}/{m_s})\partial{A_{\parallel h}}/\partial{t} = \{v_\parallel,H_h\}-({q_s}/{m_s})\partial{A_{\parallel h}}/\partial{t}$, are also continuous across cell interfaces.\footnote{In a general non-orthogonal field-aligned geometry this is not necessarily true. This is because $\v{B}^*\v{\cdot}\nabla z$ contains $A_{\parallel _h}$, which can be discontinuous in the $z$ direction. This makes the characteristic speed $\dot{\v{R}}_h\v{\cdot}\nabla z$ discontinuous across $z$ cell interfaces. This will be addressed in a separate paper dealing with non-orthogonal field-aligned geometry.}

We must also discretize the field equations. We introduce the \textit{restriction} of the phase-space mesh to configuration space, $\mathcal{T}^R$, and we  denote the configuration-space cells by $\mathcal{K}_j^R\in\mathcal{T}^R$ for $j=1,...,N_R$, where $N_R$ is the number of configuration-space cells. We also restrict $\mathcal{V}_h^p$ to configuration space as
\begin{equation}
    \mathcal{X}_h^p = \mathcal{V}_h^p \setminus \mathcal{T}^R.
\end{equation}
Further, we introduce the subset of polynomials that are piecewise continuous across configuration-space cell interfaces $\fillover{\mathcal{X}}_h^p \subset \mathcal{X}_h^p$, along with an additional subset $\dashover{\mathcal{X}}_h^p\subset \mathcal{X}_h^p$ where continuity is required in the directions perpendicular to the magnetic field, but not in the direction parallel to the field. Assuming a field-aligned coordinate system \citep[\emph{e.g.}][]{beer1995field}, we will take the perpendicular directions to be $x$ and $y$, and the parallel direction to be $z$.  

Since we require $H_h$ to be continuous across all cell interfaces, this means that we require $\phi_h$ to be continuous, \textit{i.e.} $\phi_h\in\fillover{\mathcal{X}}_h^p$. Thus to solve the Poisson equation we use the (continuous) finite-element method (FEM). While one could ensure $\phi_h$ is continuous in all directions by using a three-dimensional FEM solve, we instead use a two-dimensional FEM solve in the $x$ and $y$ directions, followed by a one-dimensional smoothing operation in the $z$ direction. That is, we first solve for $\dashover{\phi}_h\in \dashover{\mathcal{X}}^p_h$ using a two-dimensional FEM solve, and then we use a smoothing/projection operation to ensure continuity in the $z$ direction. We will denote this operation as $\phi_h = \mathcal{P}_z[\dashover{\phi}_h]$ and define it below. We can make this splitting because $\nabla_\perp$ only produces coupling in the $x$ and $y$ (perpendicular) directions. 

For the two-dimensional solve, we solve for $\dashover{\phi}_h\in \dashover{\mathcal{X}}^p_h$ by multiplying Eq. (\ref{poisson}) by a test function $\xi \in \dashover{\mathcal{X}}^p_h$ and integrating (by parts) in each configuration-space cell $\mathcal{K}^R_j$ to obtain the discrete \emph{local} weak form
\begin{equation}
\int\limits_{\mathcal{K}^R_j} d\v{R}\ \epsilon_{\perp} \nabla_\perp \dashover{\phi}\mskip0.01\thinmuskip_h \v{\cdot} \nabla_\perp \xi^{(j)} - \oint\limits_{\partial \mathcal{K}^R_j} d\v{s}_R\v{\cdot}\nabla_\perp\dashover{\phi}\mskip0.01\thinmuskip_h\ \xi^{(j)}\ \epsilon_{\perp} = \int\limits_{\mathcal{K}^R_j} d\v{R}\ \xi^{(j)}\ \mathcal{P}^*_z[\sigma_{g\,h}], \label{FEMpoisson}
\end{equation}
where $\xi^{(j)}$ denotes the restriction of $\xi$ to cell $j$ and
\begin{equation}
    \sigma_{g\,h}=\sum_s q_s \int\limits_{\mathcal{T}^v}d\v{w}\  \mathcal{J}f_{s\,h},
\end{equation}
with $\mathcal{T}^v$ the restriction of $\mathcal{T}$ to velocity space.
The global weak form is then obtained by summing Eq. (\ref{FEMpoisson}) over cells in $x$ and $y$ (but not in $z$), which results in cancellation of the surface terms at cell interfaces and leaves only a global $\partial \mathcal{T}^R$ boundary term. Note that in order to maintain energetic consistency (as we will see below), the introduction of $\mathcal{P}_z$ necessitates the modification of the right-hand side of Eq. (\ref{FEMpoisson}) with $\mathcal{P}^*_z$, the adjoint of $\mathcal{P}_z$, defined as
\begin{equation}
\int_{\mathcal{T}^R}d\v{R}\ f \mathcal{P}_z[g] = \int_{\mathcal{T}^R}d\v{R}\ \mathcal{P}^*_z[f] g.
\end{equation}

For the smoothing operation $\phi_h=\mathcal{P}_z[\dashover{\phi}_h]$, we use a one-dimensional FEM solve in the $z$ direction. This can be written as the solution $\phi_h$ of the global (in $z$) weak equality
\begin{equation}
    \int_{\mathcal{T}^z_j}d\v{R}\ \chi\ \phi_h = \int_{\mathcal{T}^z_j}d\v{R}\ \chi \dashover{\phi}\mskip0.01\thinmuskip_h, \label{smooth}
\end{equation}
where $\chi\in\widehat{\mathcal{X}}^p_h\subset\mathcal{X}^p_h$, with $\widehat{\mathcal{X}}^p_h$ a subset of the configuration-space basis where continuity is required only in the $z$ direction. Here, $\mathcal{T}^z_j$ denotes a restriction of the domain that is global in $z$ but cell-wise local in $x$ and $y$. 
We remark that using an FEM solve for this operation makes $\mathcal{P}_z$ self-adjoint, so that $\mathcal{P}^*_z=\mathcal{P}_z$. Note, however, that one could instead use a different, local smoothing operation that is not self-adjoint, so we will keep the distinction between $\mathcal{P}_z$ and $\mathcal{P}_z^*$. Also note that $\mathcal{P}_z$ is a projection operator, in that $\mathcal{P}_z[\mathcal{P}_z[\dashover{\phi}_h]]=\mathcal{P}_z[\dashover{\phi}_h]$.

The continuous discrete Hamiltonian $H_h\in\fillover{\mathcal{V}}_h^p$ is then given by
\begin{equation}
    H_h = \frac{1}{2}m{v_{\parallel\,h}^2} + \mu B_h + q \mathcal{P}_z[\dashover{\phi}\mskip0.01\thinmuskip_h],
\end{equation}
where ${v_{\parallel\,h}^2}$ is the projection of $v_\parallel^2$ onto $\fillover{\mathcal{V}}^p_h$. Note that this is only necessary when $v_\parallel^2$ is not in the basis, \emph{i.e.} when $p_v<2$, where $p_v$ is the maximum degree of the $v_\parallel$ monomials in the basis set.

For the parallel Amp\`ere equation we will take $A_{\parallel h}\in\dashover{\mathcal{X}}^p_h$ so that $A_{\parallel h}$ is continuous in $x$ and $y$ but discontinuous in $z$. Multiplying Eq. (\ref{ampere}) by a test function $\varphi\in\dashover{\mathcal{X}}^p_h$ and integrating, we can obtain the discrete weak form of this equation. The local weak form in cell $j$ is
\begin{equation}
\int\limits_{\mathcal{K}^R_j} d\v{R}\  \nabla_\perp A_{\parallel h} \v{\cdot} \nabla_\perp \varphi^{(j)} - \oint\limits_{\partial \mathcal{K}^R_j} d\v{s}_R\v{\cdot}\nabla_\perp A_{\parallel h}\ \varphi^{(j)}\  = \mu_0 \int\limits_{\mathcal{K}^R_j} d\v{R}\ \varphi^{(j)}\ J_{\parallel h}, \label{FEMampere-local}
\end{equation}
where again the surface terms will cancel on summing over cells except at the global $\partial \mathcal{T}^R$ boundary, and
\begin{equation}
J_{\parallel h} = \sum_s\frac{q_s}{m_s} \int\limits_{\mathcal{T}^v}d\v{w}\ \mathcal{J}\pderiv{H_{s\,h}}{v_\parallel} f_{s\,h}. \label{Jparh}
\end{equation}
Here, note that we have replaced the $v_\parallel$ in the $J_\parallel$ definition from Eq. (\ref{ampere}) with $(1/m)\partial H_h/\partial v_\parallel$; this will be required for energy conservation in the $p_v=1$ case, since $\partial H_h/\partial v_\parallel\neq m v_\parallel$ when $v_\parallel^2$ is not in the basis. Instead, for $p_v=1$, $\partial H_h/\partial v_\parallel=m\bar{v}_\parallel$, the piecewise-constant projection of $mv_\parallel$. As before, we solve Eq. (\ref{FEMampere-local}) using a two-dimensional FEM solve in the $x$ and $y$ directions. Note, however, that we do not require the smoothing operation in $z$ here because $A_{\parallel h}$ is allowed to be discontinuous in the $z$ direction, since it does not appear in the Hamiltonian in the symplectic formulation of EMGK.

The discrete weak form of Ohm's law can be obtained by taking the time derivative of Eq. (\ref{FEMampere-local}); after some manipulation, which we leave to Appendix \ref{app:Ohm}, the local weak form becomes
\begin{align}
    \int\limits_{\mathcal{K}^R_j}\!\!d\v{R}\, &\nabla_\perp \pderiv{A_{\parallel h}}{t} \v{\cdot} \nabla_\perp \varphi^{(j)} - \!\!\oint\limits_{\partial \mathcal{K}^R_j} \!\!\!d\v{s}_R\v{\cdot}\nabla_\perp \pderiv{A_{\parallel h}}{t}\ \varphi^{(j)}
    - \!\!\int\limits_{\mathcal{K}_j^R}\!\!d\v{R}\, \varphi^{(j)}\pderiv{A_{\parallel h}}{t} \sum_{s,i} \frac{\mu_0 q_s^2}{m_s}\!\oint\limits_{\partial\mathcal{K}^v_i} \!\!ds_w\, \bar{v}_\parallel^- \widehat{\mathcal{J} f_{s\,h}} \notag \\
    &= \mu_0\sum_s q_s \!\!\int\limits_{\mathcal{K}^R_j} \!\!d\v{R}\, \varphi^{(j)} \left[\int\limits_{\mathcal{T}^v}\!\!d\v{w}\, \bar{v}_\parallel \pderiv{(\mathcal{J}f_{s\,h})}{t}^\star-\sum_i\!\oint\limits_{\partial \mathcal{K}_i^v}\!\!ds_w\,  \bar{v}_\parallel^-{\dot{v}^H_{\parallel h}} \widehat{\mathcal{J}f_{s\,h}} \right], \qquad(p_v=1) \label{Ohmp1} \\
   \int\limits_{\mathcal{K}^R_j}\!\!d\v{R}\, &\nabla_\perp \pderiv{A_{\parallel h}}{t} \v{\cdot} \nabla_\perp \varphi^{(j)} - \!\!\oint\limits_{\partial \mathcal{K}^R_j} \!\!\!d\v{s}_R\v{\cdot}\nabla_\perp \pderiv{A_{\parallel h}}{t}\ \varphi^{(j)}
   + \!\!\int\limits_{\mathcal{K}_j^R}\!\!d\v{R}\, \varphi^{(j)}\pderiv{A_{\parallel h}}{t} \sum_s \frac{\mu_0 q_s^2}{m_s}\!\int\limits_{\mathcal{T}^v}\!\! d\v{w}\, \mathcal{J} f_{s\,h} \notag \\
    &=\mu_0\sum_s q_s \!\!\int\limits_{\mathcal{K}^R_j} \!\!d\v{R}\, \varphi^{(j)}\!\int\limits_{\mathcal{T}^v}\!\!d\v{w}\, v_\parallel \pderiv{(\mathcal{J} f_{s\,h})}{t}^\star, \qquad (p_v>1) \label{Ohmp2}
\end{align}
where $\pderivInline{A_{\parallel h}}{t}\in \dashover{\mathcal{X}}^p_h$, and
\begin{align}
    \int\limits_{\mathcal{K}_j}d\v{R}\ d\v{w}\ \psi \pderiv{(\mathcal{J}f_h)}{t}^\star = &-\oint\limits_{\partial \mathcal{K}_j} d\v{w}\ d\v{s}_R\v{\cdot}\dot{\v{R}}_h \psi^- \widehat{\mathcal{J}f_h}+\int\limits_{\mathcal{K}_j}d\v{R}\ d\v{w}\ \mathcal{J}f_h \dot{\v{R}}_h\v{\cdot} \nabla \psi \notag \\&\quad+ \int\limits_{\mathcal{K}_j}d\v{R}\ d\v{w}\ \mathcal{J}f_h \dot{v}^H_{\parallel h} \pderiv{\psi}{v_\parallel} + \int\limits_{\mathcal{K}_j}d\v{R}\ d\v{w}\ \psi\left(\mathcal{J}C[f_h] + \mathcal{J}S\right) \label{partialGK}
\end{align}
so that the gyrokinetic equation can be written as
\begin{align}
    &\int\limits_{\mathcal{K}_j}\!\!d\v{R}\, d\v{w}\, \psi \pderiv{(\mathcal{J}f_h)}{t} =\notag\\ &\quad\int\limits_{\mathcal{K}_j}\!\!d\v{R}\, d\v{w}\, \psi \pderiv{(\mathcal{J}f_h)}{t}^\star - \oint\limits_{\partial \mathcal{K}_j}\!\!d\v{R}\, ds_w\, \!\!\left(\dot{v}^H_{\parallel h}-\frac{q}{m}\pderiv{A_{\parallel h}}{t}\right) \psi^- \widehat{\mathcal{J}f_h} - \int\limits_{\mathcal{K}_j}\!\!d\v{R}\, d\v{w}\, \mathcal{J}f_h \frac{q}{m}\pderiv{A_{\parallel h}}{t} \pderiv{\psi}{v_\parallel}. \label{DGstar}
\end{align}
Note that some special attention is required to ensure that upwinding of the numerical fluxes is handled consistently in Eqs. (\ref{Ohmp1}) and (\ref{DGstar}) in the $p_v=1$ case. The upwind flow for the $v_\parallel$ surface terms is $\dot{v}^H_{\parallel h}-({q}/{m})\partial{A_{\parallel h}}/\partial{t}$; this is somewhat problematic because we cannot readily solve for $\pderivInline{A_{\parallel h}}{t}$ from Eq. (\ref{Ohmp1}) without first knowing {the upwind direction, which depends on $\pderivInline{A_{\parallel h}}{t}$}. Thus for $p_v=1$ only, we use an approximate $\widetilde{\pderivInline{A_{\parallel h}}{t}}$, calculated using Eq. (\ref{Ohmp2}) (which contains no surface term contributions), to compute the upwind direction for the $v_\parallel$ surface terms in Eqs. (\ref{Ohmp1}) and (\ref{DGstar}). (One could extend this algorithm by iterating with a new estimate of the upwind direction based on the previous estimate of $\partial A_{\parallel\, h}/\partial t$, but we leave that for future work. The present algorithm seems to work well for the cases tested so far, {and we expect that $\widetilde{\pderivInline{A_{\parallel h}}{t}}$ results in the correct upwind direction most of the time}.)

In our modal DG scheme, integrals in the above weak forms are computed analytically using a quadrature-free scheme that results in exact integrations (of the discrete integrands). (This means there are no aliasing errors, and that integration by parts operations that led to these integrals are treated exactly, for the specified discrete representation of $f_h$ and other factors in the integrand.)  This is important for ensuring the conservation properties of the scheme, since the conservation laws in the EMGK system are indirect, involving integrals of the gyrokinetic equation. The fact that integrations are exact also has important implications for the cancellation problem. Since integrals in the discrete Ohm's law are computed exactly, the discretization errors (which are solely embedded in the discrete integrands) cancel exactly, avoiding the cancellation problem. {For more details about the modal scheme, the analytical integrations and the avoidance of the cancellation problem, we have included in Appendix \ref{app:disp} a derivation of a semi-discrete Alfv\'en wave dispersion relation that results from our scheme.}

\subsection{Discrete conservation properties} \label{discons}

Now we would like to show that the discrete system (in the continuous-time limit) preserves various conservation laws of the continuous system. As with the continuous system, we will consider the conservation properties in the absence of collisions, sources and sinks, and we will assume that the boundary conditions are either periodic or that the distribution function vanishes at the boundary.

\begin{proposition}
The discrete system conserves total number of particles (the $L_1$ norm).
\end{proposition}
\begin{proof}
Taking $\psi=1$ in the discrete weak form of the gyrokinetic equation, Eq. (\ref{DGgk}), and summing over all cells, we have
\begin{gather}
    \sum_j \pderiv{}{t}\!\int\limits_{\mathcal{K}_j}\!\!d\v{R}\, d\v{w}\, \mathcal{J}f_h + \sum_j \!\oint\limits_{\partial \mathcal{K}_j}\!\! d\v{w}\, d\v{s}_R\v{\cdot}\dot{\v{R}}_h \widehat{\mathcal{J}f_h}+\sum_j \!\oint\limits_{\partial \mathcal{K}_j}\!\! d\v{R}\, ds_w \!\left(\dot{v}^H_{\parallel h}-\frac{q}{m}\pderiv{A_{\parallel h}}{t}\right) \widehat{\mathcal{J}f_h} = 0 \notag \\
    \Rightarrow \quad \pderiv{}{t}\int\limits_{\mathcal{T}}d\v{R}\ d\v{w}\ \mathcal{J}f_h 
    = 0 
\end{gather}
where the surface terms cancel exactly at cell interfaces because the integrands (both the phase-space characteristics and the numerical fluxes) are continuous across the interfaces.
\end{proof}
\begin{proposition}
The discrete system conserves a discrete total energy, $W_h = W_{H\,h} - W_{E\,h} + W_{B\,h}$, where
\begin{gather}
    W_{H\,h} = \sum_s \int_\mathcal{T} d\v{R}\ d\v{w}\  \mathcal{J}f_{s\,h} H_{s\,h}, \\
    W_{E\,h} = \sum_s \int_\mathcal{T} d\v{R}\ d\v{w}\ \frac{\epsilon_\perp}{2}|\nabla_\perp \dashover{\phi}\mskip0.01\thinmuskip_h|^2,
\end{gather}
and
\begin{gather}
    W_{B\,h} = \int_{\mathcal{T}}d\v{R}\ \frac{1}{2\mu_0}|\nabla_\perp A_{\parallel h}|^2.
\end{gather}
\end{proposition}
\begin{proof}
The proof follows from Proposition 3.2 in \citet{hakim2019discontinuous}. 
We start by calculating 
\begin{equation}
    \frac{dW_{H\,h}}{dt} = \sum_{s,j}\int\limits_{\mathcal{K}_j} d\v{R}\ d\v{w}\ H_{s\,h}\pderiv{(\mathcal{J}f_{s\,h})}{t} + \mathcal{J}f_{s\,h}\pderiv{H_{s\,h}}{t}.\label{dWHdt1}
\end{equation}
The first term can be calculated by taking $\psi=H_h$ in Eq. (\ref{DGgk}) and summing over cells and species, since $\psi\in\mathcal{V}^p_h$ and $H_h\in\fillover{\mathcal{V}}_h^p\subset \mathcal{V}_h^p$:
\begin{align}
    \sum_{s,j}&\int\limits_{\mathcal{K}_j}\!\! d\v{R}\, d\v{w}\, H_{s\,h}\pderiv{(\mathcal{J}f_{s\,h})}{t}
     + \sum_{s,j}\!\oint\limits_{\partial \mathcal{K}_j} \!\!d\v{w}\, d\v{s}_R\v{\cdot}\dot{\v{R}}_h H_{s\,h}^- \widehat{\mathcal{J}f_{s\,h}}\notag\\
     &+ \sum_{s,j}\!\oint\limits_{\partial \mathcal{K}_j}\!\! d\v{R}\, ds_w\!\left(\dot{v}^H_{\parallel h}-\frac{q_s}{m_s}\pderiv{A_{\parallel h}}{t}\right) H_{s\,h}^- \widehat{\mathcal{J}f_{s\,h}} \notag \\
     &- \sum_{s,j}\int\limits_{\mathcal{K}_j}\!\!d\v{R}\, d\v{w}\, \mathcal{J}f_{s\,h} \!\left(\dot{\v{R}}_h\v{\cdot} \nabla H_{s\,h}+\dot{v}^H_{\parallel h}\pderiv{H_{s\,h}}{v_\parallel}\right) \notag\\ 
     &+ \sum_{s,j}\int\limits_{\mathcal{K}_j}\!\!d\v{R}\, d\v{w}\, \mathcal{J}f_{s\,h} \frac{q_s}{m_s}\pderiv{A_{\parallel h}}{t} \pderiv{H_{s\,h}}{v_\parallel} = 0. 
\end{align}
Here, we see why we must require $H_h$ to be continuous; we want the surface terms to vanish, which means the integrands must be continuous across cell interfaces so that the contributions from either side of the interface cancel exactly when we sum over cells. The numerical flux $\widehat{\mathcal{J}f_h}$ is by definition continuous across the interface, and we have already noted above that the phase-space characteristics $\dot{\v{R}}_h$ and $\dot{v}^H_{\parallel h}-({q}/{m})\partial{A_{\parallel h}}/\partial{t}$ are also continuous across cell interfaces. This leaves the Hamiltonian, which we require to be continuous so that the surface terms do indeed vanish. Further, the first volume term vanishes exactly because $\dot{\v{R}}_h\v{\cdot} \nabla H_{h}+\dot{v}^H_{\parallel h}\pderivInline{H_{h}}{v_\parallel}=\{H_{h},H_{h}\}=0$ by definition of the Poisson bracket. This leaves
\begin{equation}
    \sum_{s,j}\int\limits_{\mathcal{K}_j} d\v{R}\ d\v{w}\ H_{s\,h}\pderiv{(\mathcal{J}f_{s\,h})}{t}
    = 
    -\sum_{s,j}\int\limits_{\mathcal{K}_j}d\v{R}\ d\v{w}\ \mathcal{J}f_{s\,h} \frac{q_s}{m_s}\pderiv{A_{\parallel h}}{t} \pderiv{H_{s\,h}}{v_\parallel} 
    = 
    - \int\limits_{\mathcal{T}^R} d\v{R}\ \pderiv{A_{\parallel h}}{t} J_{\parallel h}, \label{dWHdt2}
\end{equation}
where here we see why we have defined $J_{\parallel h}$ using the derivative of $H_h$ instead of $v_\parallel$, as noted after Eq. (\ref{Jparh}). We now have the desired result for this term.
For the second term in Eq. (\ref{dWHdt1}), we have
\begin{align}
    \sum_{s,j}\int\limits_{\mathcal{K}_j} d\v{R}\ d\v{w}\  \mathcal{J}f_{s\,h}\pderiv{H_{s\,h}}{t} = \sum_{s,j}\int\limits_{\mathcal{K}_j} d\v{R}\ d\v{w}\  \mathcal{J}f_{s\,h} q_s\mathcal{P}_z[\pderiv{\dashover{\phi}\mskip0.01\thinmuskip_h}{t}] = \int\limits_{\mathcal{T}^R}d\v{R}\ \sigma_{g\,h}\mathcal{P}_z[\pderiv{\dashover{\phi}\mskip0.01\thinmuskip_h}{t}]. \label{dWHdt2_h}
\end{align}
Thus we have
\begin{equation}
    \frac{dW_{H\,h}}{dt} =- \int\limits_{\mathcal{T}^R} d\v{R}\ \pderiv{A_{\parallel h}}{t} J_{\parallel h}+ \int\limits_{\mathcal{T}^R}d\v{R}\ \sigma_{g\,h}\mathcal{P}_z[\pderiv{\dashover{\phi}\mskip0.01\thinmuskip_h}{t}], \label{dWHdt_h}
\end{equation}
which is consistent with Eq. (\ref{dWHdt}).

Next, we calculate 
\begin{align}
    \frac{dW_{E\,h}}{dt} = \sum_{j} \int\limits_{\mathcal{K}^R_j} d\v{R}\ \epsilon_{\perp}\nabla_\perp\dashover{\phi}\mskip0.01\thinmuskip_h \v{\cdot} \nabla_\perp \pderiv{\dashover{\phi}\mskip0.01\thinmuskip_h}{t} = \int\limits_{\mathcal{T}^R}d\v{R}\ \mathcal{P}^*_z[\sigma_{g\,h}]\pderiv{\dashover{\phi}\mskip0.01\thinmuskip_h}{t}=\int\limits_{\mathcal{T}^R}d\v{R}\ \sigma_{g\,h}\mathcal{P}_z[\pderiv{\dashover{\phi}\mskip0.01\thinmuskip_h}{t}], \label{dWEdt_h}
\end{align}
where we have used $\xi^{(j)}=\pderivInline{\dashover{\phi}\mskip0.01\thinmuskip_h}{t}$ in Eq. (\ref{FEMpoisson}) to make the second equality, noting that the surface term vanishes upon summing over cells because $\dashover{\phi}_h\in\dashover{\mathcal{X}}^p_h$ is continuous in the perpendicular directions. Here, we see why we modified the right-hand side of Eq. (\ref{FEMpoisson}) with $\mathcal{P}_z^*$, so that the resulting term in Eq. (\ref{dWEdt_h}) matches the one in Eq. (\ref{dWHdt2_h}).

Finally, we calculate
\begin{align}
    \frac{dW_{B\,h}}{dt} = \sum_{j} \int\limits_{\mathcal{K}^R_j} d\v{R}\ \frac{1}{\mu_0}\nabla_\perp A_{\parallel h} \v{\cdot} \nabla_\perp \pderiv{A_{\parallel h}}{t}  = \int\limits_{\mathcal{T}^R}d\v{R}\ \pderiv{A_{\parallel h}}{t} J_{\parallel h},
\end{align}
where we have used $\varphi^{(j)}=({1}/{\mu_0})\partial{A_{\parallel h}}/\partial{t}$ in Eq. (\ref{FEMampere-local}) to make the second equality, again noting that the 
surface term vanishes upon summing over cells because $\pderivInline{A_{\parallel h}}{t}\in\dashover{\mathcal{X}}^p_h$ is continuous in the perpendicular directions.

We now have conservation of discrete total energy:
\begin{equation}
    \frac{dW_h}{dt} = \frac{dW_{H\,h}}{dt}-\frac{dW_{E\,h}}{dt} + \frac{dW_{B\,h}}{dt} = 0.
\end{equation}
We note that this proof did not rely on the particular choice of numerical flux function.
\end{proof}
\begin{proposition}
The discrete system exactly conserves the $L_2$ norm of the distribution function when using a central flux, while the distribution function $L_2$ norm monotonically decays when using an upwind flux.
\end{proposition}
\begin{proof}
The proof is given as Proposition 3.3 in \citet{hakim2019discontinuous}.
\end{proof}
\begin{proposition}\label{prop:entropy}
  If the discrete distribution function $f_h$ remains positive
  definite, then the discrete scheme grows the discrete entropy
  monotonically,
  \begin{align}
    -\frac{d}{dt}\int_{\mathcal{T}} d\v{R}d\v{w}\ \mathcal{J}f_h \ln(f_h) \ge 0.
  \end{align}
\end{proposition}
\begin{proof}
The proof is given as Proposition 3.4 in \citet{hakim2019discontinuous}.
\end{proof}

\section{Time-discretization scheme}\label{sec:timedisc}
So far we have considered only the discretization of the phase space for the system, and we have considered the conservation properties of the scheme in the continuous-time limit. Indeed, in the discrete-time system the conservation properties are no longer exact due to truncation error in the non-reversible time-stepping methods that we consider. However the errors will be \emph{independent} of the phase-space discretization, and errors can be reduced by taking a smaller time step or by using a high-order time-stepping scheme to improve convergence. Following the approach of the Runge-Kutta discontinuous Galerkin method \citep{cockburn1998runge,cockburn2001runge,shu2009discontinuous}, we have implemented several explicit multi-stage strong stability-preserving Runge-Kutta high-order schemes \citep{gottlieb2001strong,shu2002survey}; the results in this paper use a three-stage, third-order scheme (SSP-RK3). These schemes have the property that a high-order scheme can be composed of several forward-Euler stages. Thus we will detail our time-stepping scheme for a single forward-Euler stage, which can then be combined into a multi-stage high-order scheme. Note that although we present the time-discretization scheme in this section in terms of our DG phase-space discretization, the scheme could be generalized to any spatial discretization.

Given $f_h^n=f_h(t=t_n)$ and $A_{\parallel h}^n=A_{\parallel h}(t=t_n)$ at time $t_n$, the steps of the forward-Euler scheme to advance to time $t_{n+1}=t_n+\Delta t$ are as follows:
\begin{enumerate}
    \item Calculate $\dashover{\phi}_h^n$ using Eq. (\ref{FEMpoisson}), and then $\phi_h^n = \mathcal{P}_z[\dashover{\phi}_h^n]$ using Eq. (\ref{smooth}).
    \begin{gather}
\int\limits_{\mathcal{K}^R_j} d\v{R}\ \epsilon_{\perp} \nabla_\perp \dashover{\phi}\mskip0.01\thinmuskip_h^n \v{\cdot} \nabla_\perp \xi^{(j)} - \oint\limits_{\partial \mathcal{K}^R_j} d\v{s}_R\v{\cdot}\nabla_\perp\dashover{\phi}\mskip0.01\thinmuskip_h^n\ \xi^{(j)}\ \epsilon_{\perp} = \int\limits_{\mathcal{K}^R_j} d\v{R}\ \xi^{(j)}\ \mathcal{P}^*_z[\sigma_{g\,h}^n] \label{FEMpoisson-local-dt} \\
    \int_{\mathcal{T}^z_j}d\v{R}\ \chi\ \phi_h^n = \int_{\mathcal{T}^z_j}d\v{R}\ \chi \dashover{\phi}\mskip0.01\thinmuskip_h^n
\end{gather}
    \item Calculate the partial EMGK update $\left({\partial(\mathcal{J}f_h)}^\star/{\partial t}\right)^n$ using Eq. (\ref{partialGK}).
\begin{align}
    \int\limits_{\mathcal{K}_j}d\v{R}\ d\v{w}\ \psi \left(\pderiv{(\mathcal{J}f_h)}{t}^\star\right)^n = &-\oint\limits_{\partial \mathcal{K}_j} d\v{w}\ d\v{s}_R\v{\cdot}\dot{\v{R}}_h^n \psi^- \widehat{\mathcal{J}f_h}^n+\int\limits_{\mathcal{K}_j}d\v{R}\ d\v{w}\ \mathcal{J}f_h^n \dot{\v{R}}_h^n\v{\cdot} \nabla \psi \notag \\&\quad+ \int\limits_{\mathcal{K}_j}d\v{R}\ d\v{w}\ \mathcal{J}f_h^n \dot{v}^{H\,n}_{\parallel h} \pderiv{\psi}{v_\parallel} + \int\limits_{\mathcal{K}_j}d\v{R}\ d\v{w}\ \psi\left(\mathcal{J}C[f_h^n] + \mathcal{J}S^n\right) \label{partialGK-dt}
\end{align}
    \item Calculate $\left({\partial A_{\parallel h}}/{\partial t}\right)^n$ from Eq. (\ref{Ohmp2}) [\textit{for $p_v=1$, this is only a provisional value, which we will denote as $(\widetilde{\partial{A_{\parallel h}}/\partial{t}})^n$}].
\begin{align}
        \int\limits_{\mathcal{K}^R_j} d\v{R}\  &\nabla_\perp \left(\pderiv{A_{\parallel h}}{t}\right)^n \v{\cdot} \nabla_\perp \varphi^{(j)} - \oint\limits_{\partial \mathcal{K}^R_j} d\v{s}_R\v{\cdot}\nabla_\perp \left(\pderiv{A_{\parallel h}}{t}\right)^n \varphi^{(j)} \notag \\
        &+
    \int\limits_{\mathcal{K}_j^R}d\v{R}\ \varphi^{(j)}\left(\pderiv{A_{\parallel h}}{t}\right)^n \sum_s \frac{\mu_0 q_s^2}{m_s}\int\limits_{\mathcal{T}^v} d\v{w}\ \mathcal{J} f_{s\,h}^n \notag \\
    &=\mu_0\sum_s q_s \int\limits_{\mathcal{K}_j^R} d\v{R}\ \varphi^{(j)}\int\limits_{\mathcal{T}^v}d\v{w}\ v_\parallel \left(\pderiv{(\mathcal{J} f_{s\,h})}{t}^\star\right)^n
\end{align}
    \item ($p_v=1\ only$) Use the provisional $(\widetilde{\partial{A_{\parallel h}}/\partial{t}})^n$ from step 3 to calculate the upwinding direction in the surface terms in Eq. (\ref{Ohmp1}), and then calculate $({\partial{A_{\parallel h}}/\partial{t}})^n$.
\begin{align}
        \int\limits_{\mathcal{K}^R_j} d\v{R}\  &\nabla_\perp \left(\pderiv{A_{\parallel h}}{t}\right)^n \v{\cdot} \nabla_\perp \varphi^{(j)} - \oint\limits_{\partial \mathcal{K}^R_j} d\v{s}_R\v{\cdot}\nabla_\perp \left(\pderiv{A_{\parallel h}}{t}\right)^n \varphi^{(j)} \notag \\
    &- \int\limits_{\mathcal{K}_j^R}d\v{R}\ \varphi^{(j)}\left(\pderiv{A_{\parallel h}}{t}\right)^n \sum_s \frac{\mu_0 q_s^2}{m_s}\sum_i\oint\limits_{\partial\mathcal{K}^v_i} ds_w \bar{v}_\parallel^- \widehat{\mathcal{J} f_{s\,h}}^n \notag \\
    &= \mu_0\sum_s q_s \int\limits_{\mathcal{K}^R_j} d\v{R}\ \varphi^{(j)} \left[\int\limits_{\mathcal{T}^v}d\v{w}\ \bar{v}_\parallel \left(\pderiv{(\mathcal{J}f_{s\,h})}{t}^\star\right)^n-\sum_i\oint\limits_{\partial \mathcal{K}_i^v} ds_w\  \bar{v}_\parallel^- {\dot{v}^{H\,n}_{\parallel h}} \widehat{\mathcal{J}f_{s\,h}}^n \right]
\end{align}
    \item Calculate the full EMGK update, $\left(\pderiv{(\mathcal{J}f_h)}{t}\right)^{n}$, using Eq. (\ref{DGstar}). For $p_v=1$, the provisional $\left(\widetilde{\pderiv{A_{\parallel h}}{t}}\right)^n$ from step 3 should again be used to calculate the upwinding direction in the surface terms for consistency.
\begin{align}
    &\int\limits_{\mathcal{K}_j}\!\!d\v{R}\, d\v{w}\, \psi
    \left(\pderiv{(\mathcal{J}f_h)}{t}\right)^n = \int\limits_{\mathcal{K}_j}\!\!d\v{R}\, d\v{w}\, \psi \left(\pderiv{(\mathcal{J}f_h)}{t}^\star\right)^n \notag \\
    &\qquad- \oint\limits_{\partial \mathcal{K}_j}\!\! d\v{R}\, ds_w\!\left[\dot{v}^{H\,n}_{\parallel h}-\frac{q}{m}\left(\pderiv{A_{\parallel h}}{t}\right)^n\right] \psi^- \widehat{\mathcal{J}f_h}^n
    - \int\limits_{\mathcal{K}_j}\!\!d\v{R}\, d\v{w}\, \mathcal{J}f_h^n \frac{q}{m}\left(\pderiv{A_{\parallel h}}{t}\right)^n \pderiv{\psi}{v_\parallel}. \label{DGstar-dt}
\end{align}
    \item Advance $f_h$ and $A_{\parallel h}$ to time $t_{n+1}$.
\begin{gather}
   \mathcal{J} f_h^{n+1} = \mathcal{J} f_h^n + \Delta t \left(\pderiv{(\mathcal{J} f_h)}{t}\right)^n \label{f-dt} \\
   A_{\parallel h}^{n+1} = A_{\parallel h}^n + \Delta t \left(\pderiv{A_{\parallel h}}{t}\right)^n \label{Apar-dt}
\end{gather}
\end{enumerate}
Note that the parallel Amp\`ere equation, Eq. (\ref{FEMampere-local}), is only used to solve for the initial condition of $A_{\parallel h}(t=0)$. For all other times, Eq. (\ref{Apar-dt}) is used to advance $A_{\parallel h}$. This prevents the system from being over-determined and ensures consistency between $A_{\parallel h}$ and $\pderivInline{A_{\parallel h}}{t}$.

\section{Linear benchmarks}
\label{sec:linear}

\subsection{Kinetic Alfv\'en wave}\label{res:kaw}
As a first benchmark of our electromagnetic scheme, we consider the kinetic Alfv\'en wave.
In a slab (straight background magnetic field) geometry, with stationary ions (assuming $\omega\gg k_\parallel v_{ti}$), the gyrokinetic equation for electrons reduces to
\begin{equation}
    \pderiv{f_e}{t} = \{H_e,f_e\}-\frac{e}{m}\pderiv{f_e}{v_\parallel}\pderiv{A_\parallel}{t} = - v_\parallel \pderiv{f_e}{z} {-} \frac{e}{m}\pderiv{f_e}{v_\parallel}\left(\pderiv{\phi}{z}+\pderiv{A_\parallel}{t} \right).
\end{equation}
Taking a single Fourier mode with perpendicular wavenumber $k_\perp$ and parallel wavenumber $k_\parallel$, the field equations become
\begin{gather}
    k_\perp^2 \frac{m_i n_{0}}{B^2}\phi = e n_0-e\int dv_\parallel\ f_e \label{alf-poisson} \\
    k_\perp^2 A_\parallel = -\mu_0 e \int dv_\parallel\ v_\parallel f_e \label{alf-ampere} \\
    \left(k_\perp^2 + \frac{\mu_0 e^2}{m_e}\int dv_\parallel f_e\right)\pderiv{A_{\parallel}}{t} = -\mu_0 e\int dv_\parallel\ v_\parallel \{H_e, f_e\}
\end{gather}
After linearizing the gyrokinetic equation by assuming a uniform Maxwellian background with density $n_0$ and temperature $m_e v_{te}^2$, so that $f_e = F_{Me} + \delta f_e$, the dispersion relation becomes
\begin{equation}
    \omega^2\left[1+\frac{\omega}{\sqrt{2}k_\parallel v_{te}}Z\left(\frac{\omega}{\sqrt{2}k_\parallel v_{te}}\right)\right] = \frac{k_\parallel^2 v_{te}^2}{\hat{\beta}}\left[1+k_\perp^2\rho_s^2+\frac{\omega}{\sqrt{2}k_\parallel v_{te}}Z\left(\frac{\omega}{\sqrt{2}k_\parallel v_{te}}\right)\right], \label{kawdisp}
\end{equation}
where $\hat{\beta}=(\beta_e/2) m_i/m_e$, with $\beta_e = 2\mu_0 n_0 T_e/B^2$, $v_{te}=\sqrt{T_e/m_e}$ is the electron thermal speed, $\rho_s$ is the ion sound gyroradius, and $Z(x)$ is the plasma dispersion function \citep{fried1961plasma}. In the limit $k_\perp\rho_s\ll1$ the wave becomes the standard shear Alfv\'en wave from magnetohydrodynamics (MHD), which is an undamped wave with frequency $\omega=k_\parallel v_A$, where $v_A=v_{te}/\hat{\beta}^{1/2}$ is the Alfv\'en velocity. For larger values of $k_\perp \rho_s$, the mode is damped by kinetic effects.

\begin{figure}
    \hspace{-.5cm}
    \centering
    \includegraphics[width=.48\textwidth]{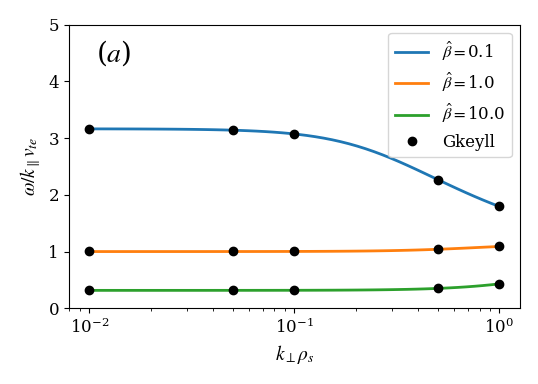}
    \includegraphics[width=.48\textwidth]{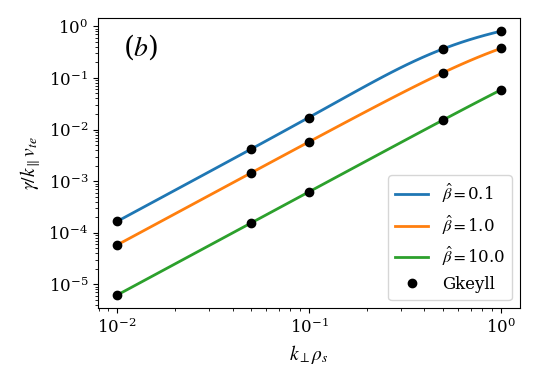}
    \caption{Real frequencies ($a$) and damping rates ($b$) for the kinetic Alfv\'en wave vs $k_\perp \rho_s$. Solid lines are the exact values from Eq. (\ref{kawdisp}) for three different values of $\hat{\beta}=(\beta_e/2) m_i/m_e$, and black dots are the numerical results from \gke.}
    \label{fig:kaw}
\end{figure}

\begin{figure}
    \centering
    \includegraphics[width=.6\textwidth]{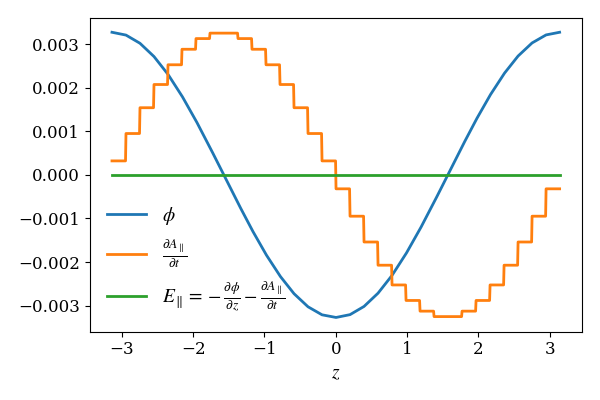}
    \caption{$\phi_h$ (blue) and $\partial{A_{\parallel h}}/\partial {t}$ (yellow) for the case with $\hat{\beta}=10$ and $k_\perp \rho_s=0.01$. The amplitude of $E_{\parallel h}$ (green) is $\sim 10^{-9}$.}
    \label{fig:kaw-fields}
\end{figure}

In Figure \ref{fig:kaw}, we show the real frequencies ($a$) and damping rates ($b$) obtained by solving Eq. (\ref{kawdisp}) for a few values of $\hat{\beta}$. We also show numerical results from \gke, which match the analytic results very well. These results are a good indication that our scheme avoids the Amp\`ere cancellation problem, which can cause large errors for modes with $\hat{\beta}/k_\perp^2\rho_s^2\gg1$ (see Appendix \ref{app:cancel}); we see no such errors, even for the case with $\hat{\beta}/k_\perp^2\rho_s^2=10^5$. Each \gke simulation was run using piecewise-linear basis functions ($p=1$) in a reduced dimensionality mode with one configuration space dimension and one velocity space dimension, with $(N_z,N_{v_\parallel})=(32,64)$ the number of cells in each dimension. The perpendicular dimensions ($x$ and $y$), which appear only in the field equations in this simple system, were handled by replacing $\nabla_\perp$ by $k_\perp$, as in Eqs. (\ref{alf-poisson}) and (\ref{alf-ampere}). We use periodic boundary conditions in $z$ and zero-flux boundary conditions in $v_\parallel$.

We also show in Figure \ref{fig:kaw-fields} the fields $\phi_h$ and $\pderivInline{A_{\parallel h}}{t}$ for the case with $\hat{\beta}=10$ and $k_\perp \rho_s=0.01$, which gives $\hat{\beta}/k_\perp^2\rho_s^2=10^5$. For these parameters the system is near the MHD limit, which means we should expect $E_\parallel=-\pderivInline{\phi}{z}-\pderivInline{A_\parallel}{t}\approx0$. While this condition is never enforced, getting the physics correct requires the scheme to allow $\pderivInline{\phi_h}{z}\approx-\pderivInline{A_{\parallel h}}{t}$. The fact that our scheme allows discontinuities in $A_\parallel$ in the parallel direction is an advantage in this case. Because $\phi_h$ is piecewise-linear here, $\pderivInline{\phi_h}{z}$ is piecewise-constant; this is necessarily discontinuous for non-trivial solutions. Thus the scheme produces a piecewise-constant $\pderivInline{A_{\parallel h}}{t}$ in this MHD-limit case, as shown in Figure \ref{fig:kaw-fields}, resulting in $E_{\parallel h}\approx 0$.
If our scheme did not allow discontinuities in $A_{\parallel h}$, a continuous $\pderivInline{A_{\parallel h}}{t}$ would never be able to exactly cancel a discontinuous $\pderivInline{\phi_h}{z}$, and the resulting $E_{\parallel h}\neq0$ would make the solution inaccurate. Notably, this would be the case had we chosen the Hamiltonian ($p_\parallel$) formulation of the gyrokinetic system, which uses $p_\parallel=mv_\parallel + q A_\parallel$ as the parallel velocity coordinate. This is because $A_\parallel$ is included in the Hamiltonian in the $p_\parallel$ formulation, which would require continuity of $A_{\parallel h}$ (and thereby $\pderivInline{A_{\parallel h}}{t}$) to conserve energy in our discretization scheme. 

\subsection{Kinetic ballooning mode (KBM)}
We use the kinetic ballooning mode (KBM) instability in the local limit as a second linear benchmark of our electromagnetic scheme. The dispersion relation is given by solving \citep{kim1993electromagnetic}
\begin{gather}
    \omega\left[\tau + k_\perp^2 + \Gamma_0(b) - P_0\right]\phi = \left[\tau(\omega-\omega_{*e}) - k_\parallel P_1\right]\frac{\omega}{k_\parallel}A_\parallel \label{kbmdisp1}\\
    \frac{2 k_\parallel^2 k_\perp^2}{\beta_i}A_\parallel = k_\parallel\left[k_\parallel P_1 - \tau(\omega-\omega_{*e})\right]\phi - \left[k_\parallel^2P_2 - \tau\left(\omega(\omega-\omega_{*e})-2\omega_{de}(\omega-\omega_{*e}(1+\eta_e))\right)\right]A_\parallel \label{kbmdisp2}
\end{gather}
where
\begin{gather}
    P_m = \int_0^{\infty}dv_\perp \ v_\perp\int_{-\infty}^{\infty} dv_\parallel\  \frac{1}{\sqrt{2\pi}}e^{-(v_\parallel^2+v_\perp^2)/2}(v_\parallel)^m\frac{\omega-\omega_{*i}\left[1+\eta_i(v^2/2-3/2)\right]}{\omega-k_\parallel v_\parallel - \omega_{di}(v_\parallel^2+v_\perp^2/2)}J_0^2(v_\perp\sqrt{b}),
\end{gather}
with $\tau=T_i/T_e$, $\omega_{*e}=k_y$, $\omega_{*i}=-k_y$, $\eta_s=L_{n}/L_{Ts}$ and $\Gamma_0(b)=I_0(b)e^{-b}$ with $I_0(b)=J_0(i b)$ the modified Bessel function. Here, the wavenumbers $k_y$ and $k_\parallel$ are normalized to $\rho_i$ and $L_n$, respectively, and the frequencies $\omega$ and $\omega_*$ are normalized to $v_{ti}/L_n$. In the local limit, $\omega_{ds}=\omega_{*s}L_{n}/R$ and $k_\perp=k_y$ do not vary along the field line. Note that in Eq. (\ref{kbmdisp1}) we have modified the FLR terms from \citep{kim1993electromagnetic} so that we can take $b=k_\perp^2\rightarrow0$ while keeping $k_\perp\neq0$ to neglect all FLR effects except for the first-order polarization term, which is consistent with our long-wavelength Poisson equation. 

The local limit can be achieved by simulating a helical flux tube with no magnetic shear, which gives a system with constant magnetic curvature that corresponds to $\omega_d=\text{const}$. This geometry has been previously used for SOL turbulence studies with \gke \citep{shi2019full,bernard2019gyrokinetic}, except in this section we take the boundary condition along the field lines to be periodic. We will provide further details about the helical geometry and the coordinates in the next section.

\begin{figure}
    \centering
    \includegraphics[width=.7\textwidth]{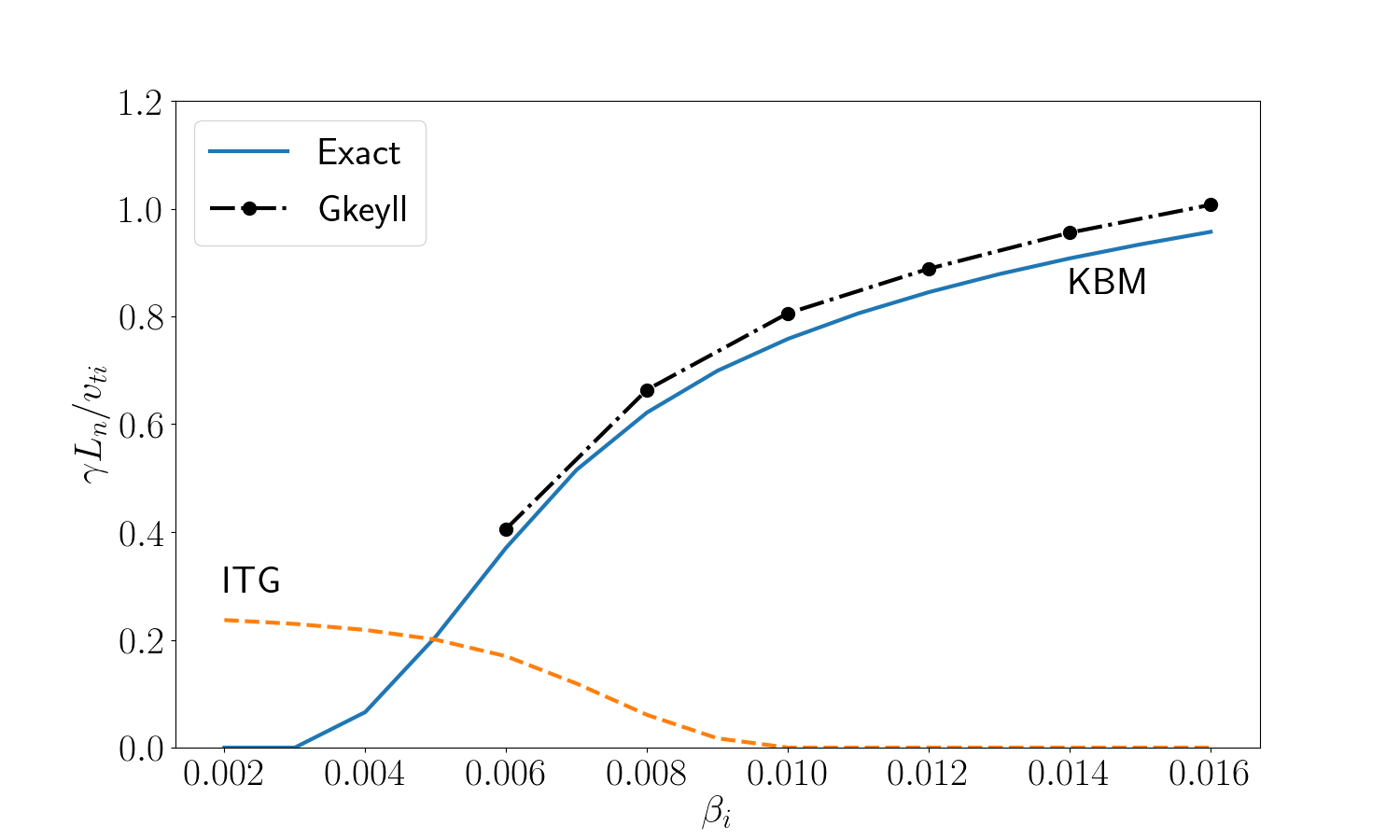}
    \caption{Growth rates for the KBM instability in the local limit, as a function of $\beta_i$, with $k_\perp \rho_i=0.5,\ k_\parallel L_n = 0.1,\ R/L_n=5,\ R/L_{Ti}=12.5,\  R/L_{Te}=10,$ and $\tau=1$. The black dots are numerical results from \gke, and the colored lines are the result of numerically solving the analytic dispersion relation given by Eqs. (\ref{kbmdisp1}-\ref{kbmdisp2}).}
    \label{fig:kbm}
\end{figure}

We show the results of \gke simulations of the KBM instability in the local-limit helical geometry for several values of $\beta_i$ in Figure \ref{fig:kbm}. The results agree well with the analytic result obtained by numerically solving Eqs. (\ref{kbmdisp1}-\ref{kbmdisp2}). The parameters $k_\perp \rho_i=0.5,\ k_\parallel L_n = 0.1,\ R/L_n=5,\ R/L_{Ti}=12.5,\  R/L_{Te}=10,\ \tau=1$ are chosen to match those used in figure 1 of \citet{kim1993electromagnetic}, although the differences in FLR terms ($b=0$) cause our growth rates to be larger than those in \citet{kim1993electromagnetic}. These are fully five-dimensional simulations with the real deuterium--electron mass ratio using piecewise-linear ($p=1$) basis functions, with $(N_x,N_y,N_z,N_{v_\parallel},N_\mu)=(1,16,16,32,16)$. The boundary conditions are periodic in the three configuration-space dimensions and zero flux in the velocity dimensions. The initial distribution function of each species is composed of a background Maxwellian with gradients in the density and temperature corresponding to the desired $L_n$ and $L_{Ts}$, plus a perturbed Maxwellian (for the electrons only) with small sinusoidal variations in the density corresponding to the desired $k_y$ and $k_\parallel$.
Note that since we are using a full-$f$ representation, the presence of a background gradient in the distribution function means that we must apply the periodic boundary conditions by first subtracting off the initial background distribution function, then applying periodicity to the perturbations only and then adding back the background distribution.
{Finally, we note that since \gke is designed primarily for nonlinear calculations, the fact that Fourier modes are not eigenfunctions of the DG discretization of the system makes these linear tests somewhat difficult for \gke. This may play a role in the small deviation of the results from the analytical theory. Because of this, Fourier modes other than the one initialized can grow and pollute the results. In particular, we have not included results from the ion temperature gradient (ITG) branch because we find that a mode with $k_\parallel=0$ grows and overcomes the initialized finite $k_\parallel$ mode before its growth rate has converged.}

\section{Nonlinear results}
\label{sec:nonlinear}

We now present preliminary nonlinear electromagnetic results from \gke. We simulate turbulence on helical, open field lines as a rough model of the tokamak scrape-off layer. 
These simulations are a direct extension of the work of \citet{shi2019full} to include electromagnetic fluctuations. As such, we use the same simulation geometry and similar NSTX-like parameters. In the non-orthogonal field-aligned geometry, $x$ is the radial coordinate, $z$ is the coordinate along the field lines, and $y$ is the binormal coordinate which labels field lines at constant $x$ and $z$. These coordinates map to physical cylindrical coordinates ($R,\varphi,Z)$ via $R=x$, $\varphi=(y/\sin\theta+z\cos\theta)/R_c$, $Z=z\sin\theta$. (Note that this fixes an error in the $\varphi(y,z)$ mapping in \citep{shi2019full}.) The field-line pitch $\sin\theta=B_v/B$ is taken to be constant, with $B_v$ the vertical component of the magnetic field (analogous to the poloidal field in typical tokamak geometry), and $B$ the total magnitude of the background magnetic field. Further, $R_c=R_0+a$ is the radius of curvature at the centre of the simulation domain, with $R_0$ the device major radius and $a$ the minor radius. As in \citet{shi2019full}, we neglect all geometrical factors arising from the non-orthogonal coordinate system, except for the assumption that perpendicular gradients are much stronger than parallel gradients so that we can approximate
\begin{equation}
    (\nabla\times\uv{b})\v{\cdot}\nabla f(x,y,z)\approx \left[(\nabla\times\uv{b})\v{\cdot} \nabla y\right]\pderiv{f}{y} = -\frac{1}{x}\pderiv{f}{y},
\end{equation}
where we have used $\v{B}=B_\text{axis}(R_0/R)\uv{e}_z$ in the last step, with $B_\text{axis}$ the magnetic field strength at the magnetic axis.

We use this geometry to simulate a {flux-tube-like domain} on the outboard side that wraps {helically around the torus} and terminates on conducting plates at each end in $z$. The simulation box is centreed at $(x,y,z)=(R_c,0,0)$ with dimensions $L_x=50\rho_{s0}\approx 14.6$ cm, $L_y=100\rho_{s0}\approx 29.1$ cm, and $L_z=L_p/\sin\theta=8$ m, where $L_p=2.4$ m and $\rho_{s0}=c_{s0}/\Omega_i$. Periodic boundary conditions are used in the $y$ direction, and a Dirichlet boundary condition $\phi=0$ is applied in x, which effectively prevents flows into the boundaries in $x$. Conducting-sheath boundary conditions are applied in the $z$ direction \citep{shi2017gyrokinetic,shi2019full}, which partially reflect one species (typically electrons) and fully absorbs the other species depending on the sign of the sheath potential. This involves solving the gyrokinetic Poisson equation to evaluate the potential at the $z$ boundary, corresponding to the sheath entrance, and using the resulting sheath potential to determine a cutoff velocity below which particles are reflected by the sheath. {Notably, our sheath boundary condition allows current fluctuations in and out of the sheath, which we will find to important later in this section. This is different from the standard logical sheath boundary condition \citep{parker1993suitable} that imposes that there is no net current to the sheath by assuming that the ion and electron currents at the sheath entrance are equal at all times.} The velocity-space grid has extents $-4v_{ts}\leq v_\parallel \leq 4 v_{ts}$ and $0\leq\mu\leq6T_{s0}/B_0$, where $v_{ts}=\sqrt{T_{s0}/m_s}$ and $B_0=B_\text{axis}R_0/R_c$. We use piecewise-linear ($p=1$) basis functions, with $(N_x,N_y,N_z,N_{v_\parallel},N_\mu)=(16,32,10,10,5)$. {Note that although the domain that we simulate is a flux tube, the simulations are not performed in the local limit; the simulations include radial variation of the magnetic field and the profiles, and are thus effectively global.}

\begin{figure}
    \centering
    \includegraphics[width=\textwidth]{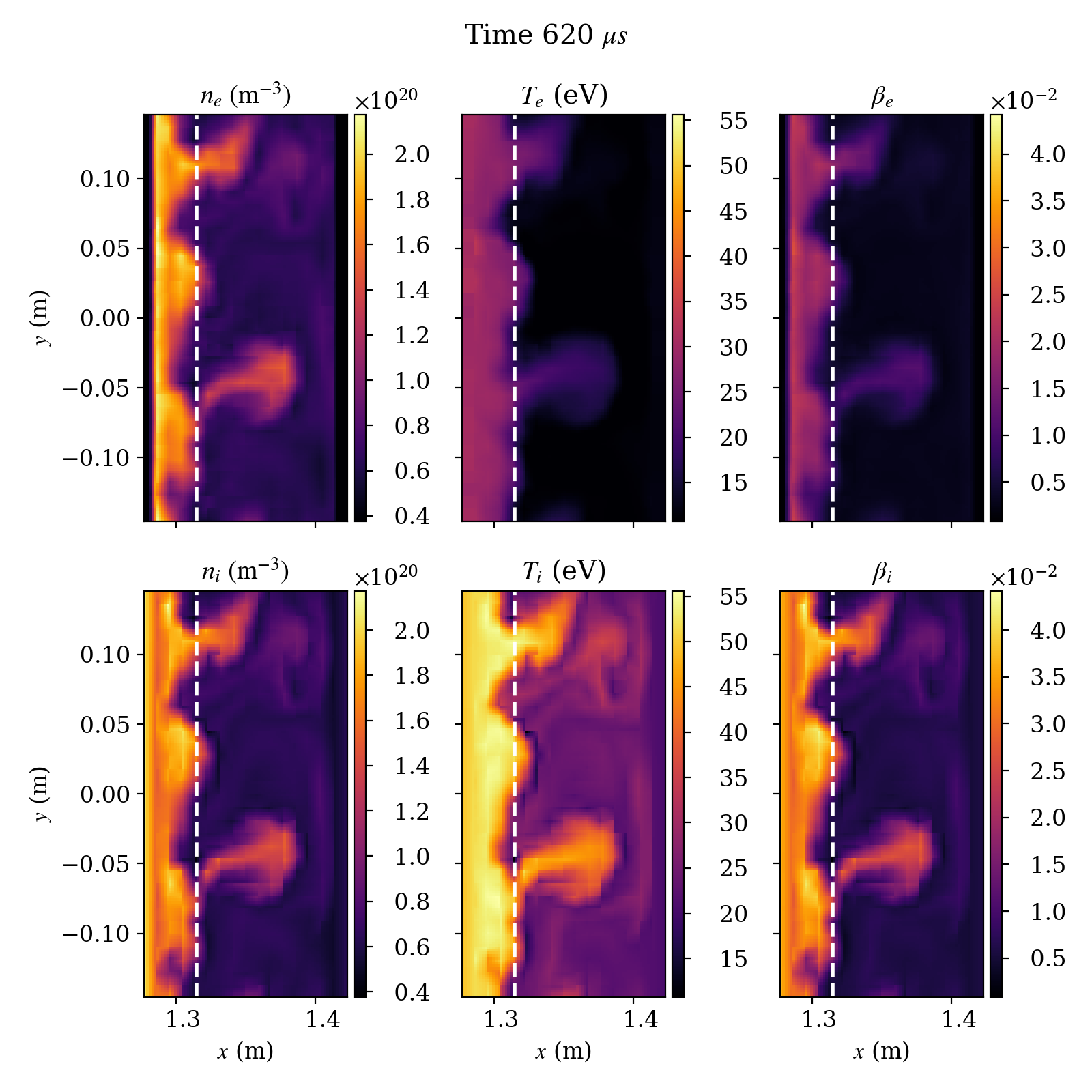}
    \caption{Snapshots from an electromagnetic simulation on open, helical field lines. From left to right, we show the density, temperature, and plasma beta of electrons (top row) and ions (bottom row). The snapshots are taken at the midplane $(z=0)$ at $t=620\ \mu$s. The dashed line indicates the
boundary between the source and SOL regions. A blob with mushroom structure is being ejected from the source region and propagating radially outward into the SOL region.}
    \label{fig:moms}
\end{figure}

\begin{figure}
    \centering
    \includegraphics[width=\textwidth]{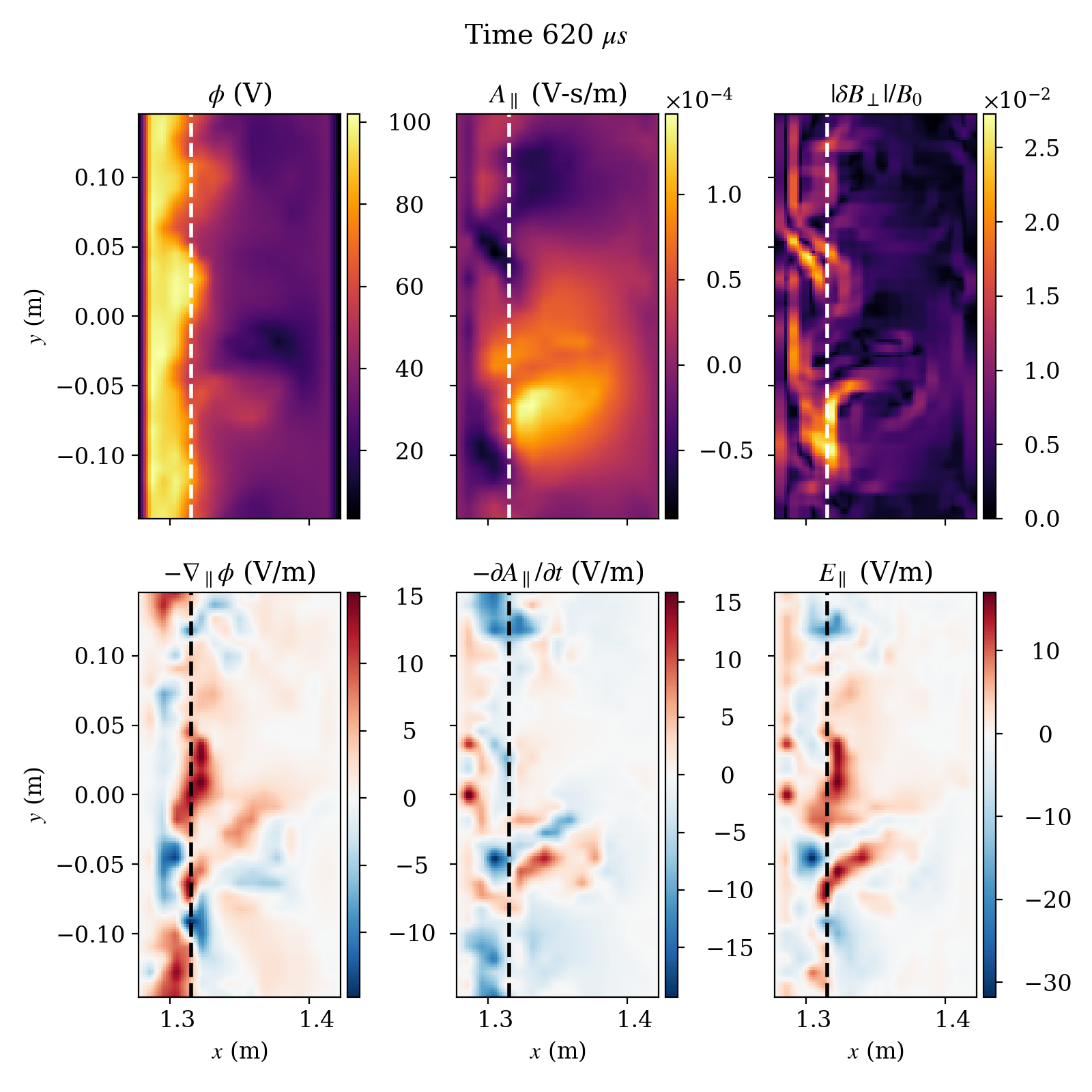}
    \caption{Snapshots (at $z=0$, $t=620\ \mu$s) of the electrostatic potential $\phi$, parallel magnetic vector potential $A_\parallel$, and normalized magnetic fluctuation amplitude $|\delta B_\perp|/B_0=|\nabla_\perp A_\parallel|/B_0$ (top row), along with the components of the parallel electric field $E_\parallel = -\nabla_\parallel \phi - \partial{A_\parallel}/\partial{t}$ (bottom row).}
    \label{fig:fields}
   \centering
    \includegraphics[width=.5\textwidth]{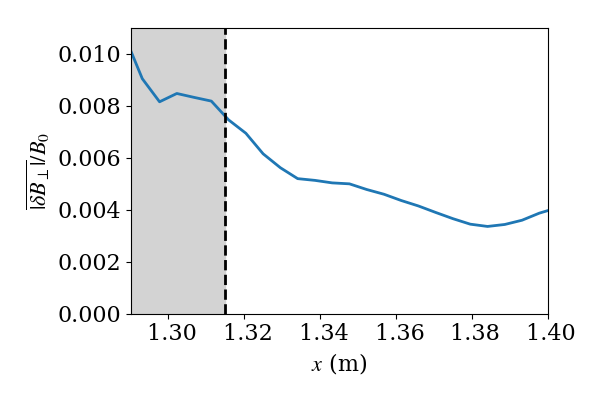}
    \caption{Radial profile of the normalized magnetic fluctuation amplitude, $|\delta B_\perp|/B_0=|\nabla_\perp A_\parallel|/B_0$, averaged in $y$, $z$ and time using data near the midplane $(|z|<0.4\ \mathrm{ m})$ over a period of 400 $\mu$s. On average, we observe magnetic fluctuations of the order of $0.5-1\%$. The source region is shaded.}
    \label{fig:deltaB-profile} 
\end{figure}

The simulation parameters are similar to those used in \citet{shi2019full}, roughly approximating an H-mode deuterium plasma in the NSTX SOL: $B_\text{axis}=0.5$ T, $R_0=0.85$ m, $a=0.5$ m, $T_{e0}=T_{i0}=40$ eV. For the particle source, we use the same form as in \citet{shi2019full} but we increase the source particle rate by a factor of 10 to access a higher $\beta$ regime where electromagnetic effects will be more important. The source is localized in the region $x<x_S +3\lambda_S$, with $x_S=R_c-0.05$ m and $\lambda_S=5\times10^{-3}$ m. The location $x=x_S +3\lambda_S$, which separates the source region from the SOL region, can be thought of as the separatrix. A floor of one tenth the peak particle source rate is used near the midplane to prevent regions of $n\ll n_0$ from developing at large $x$. Unlike in \citet{shi2019full} we do not use numerical heating to keep $f>0$ despite the fact that our DG algorithm does not guarantee positivity. While the simulations appear to be robust to negative $f$ in some isolated regions, lowering the source floor in the SOL region leads to simulation failures due to positivity issues at large $x$. A more sophisticated algorithm for ensuring positivity is left to future work.
We also artificially lower the collision frequency to one tenth the physical value to offset the increased particle source rate so that the time-step limit from collisions does not become too restrictive. Further, we model only ion--ion and electron--electron collisions, leaving cross-species collisions to future work. 

Using the novel electromagnetic scheme described in this paper, we ran a simulation in this configuration 
to $t=1$ ms, with a quasi-steady state being reached around $t=600\ \mu\text{s}$ when the sources balance losses to the end plates. For reference, the ion transit time is $\tau_i=(L_z/2)/v_{ti}\approx 50\ \mu\text{s}$. In Figure \ref{fig:moms} we show snapshots of the density, temperature and $\beta$ of electrons (top row) and ions (bottom row). The snapshots are taken at the midplane ($z=0$) at $t=620\ \mu$s. We can see a blob with a mushroom structure being ejected from the source region. We also show in Figure \ref{fig:fields} snapshots of the electromagnetic fields taken at the same time and location. We show the electrostatic potential $\phi$, the parallel magnetic vector potential $A_\parallel$ and the normalized magnetic fluctuation amplitude $|\delta B_\perp|/B_0=|\nabla_\perp A_\parallel|/B_0$ (top row), along with the components of the parallel electric field $E_\parallel = -\nabla_\parallel \phi - \pderivInline{A_\parallel}{t}$ (bottom row). Note that only $\phi$, $A_\parallel$ and $\pderivInline{A_\parallel}{t}$ are evolved quantities in the simulation, with the other quantities derived. We see that $\pderivInline{A_\parallel}{t}$ is of comparable magnitude to $\nabla_\parallel \phi$, indicating that the dynamics is in the electromagnetic regime. Significant magnetic fluctuations of over $2.5\%$ can be seen in $|\delta B_\perp|/B_0$ in this snapshot. We also show in Figure \ref{fig:deltaB-profile} the time and spatially averaged profile of magnetic fluctuations vs $x$, which shows that, on average, we observe magnetic fluctuations of the order of $0.5-1\%$. This radial profile, and similar ones that will follow, is computed by averaging in $y$ and $z$ using data near the midplane $(|z|<0.4\ \mathrm{ m})$ over the period of $600\ \mu\text{s}-1$ ms.

\begin{figure}
    \centering
    \includegraphics[width=\textwidth]{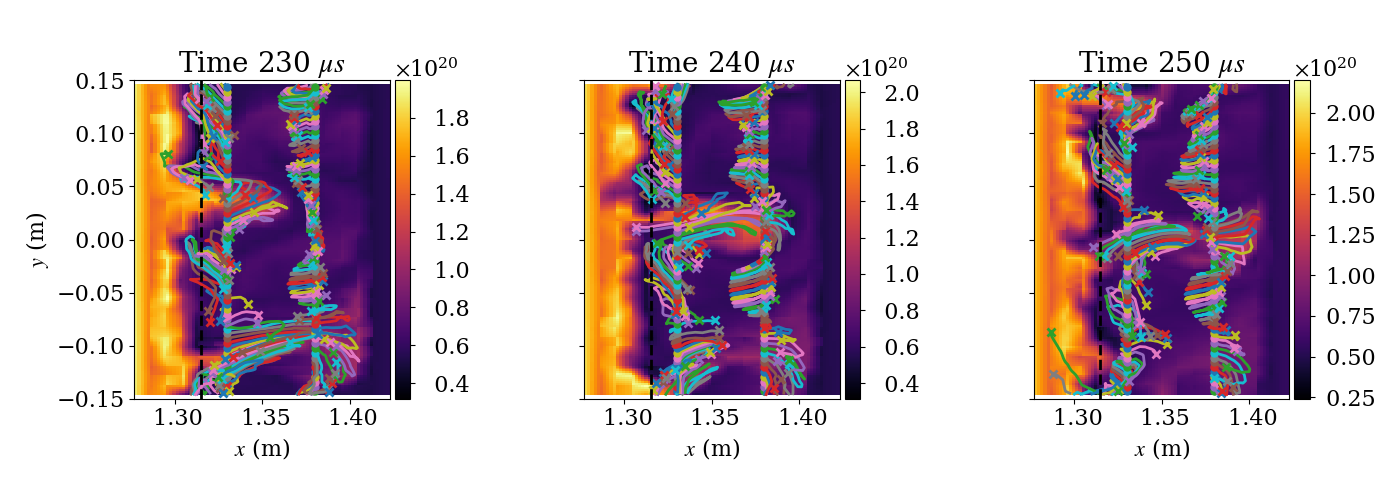}
    \caption{Three-dimensional magnetic field line trajectories at $t=230$, $240$, and $250$~$\mu$s, projected onto the $x-y$ plane.  The ion density at $z=0$~m is plotted in the background.  Each field line starts at $z=-4$~m and either $x=1.33$~m or $x=1$.38~m for a range of $y$v alues and is traced to $z=4$~m. The starting points are marked with circles and the ending points are marked with crosses. Focusing on the blob that is being ejected near $y=0$~m, we see that field lines are stretched and bent by the blob as it propagates into the SOL region. In previous frames (not shown) a blob was also ejected near $y=-0.1$~m. At $t=230$~$\mu$s the field lines are still stretched from this event, but they return closer to their equilibrium position by $t=250$~$\mu$s. (A full animation of this timeseries is included as supplementary materials online at https://doi.org/10.1017/S0022377820000070.)} 
    \label{fig:bstream-y}
    \centering
    \includegraphics[width=.9\textwidth]{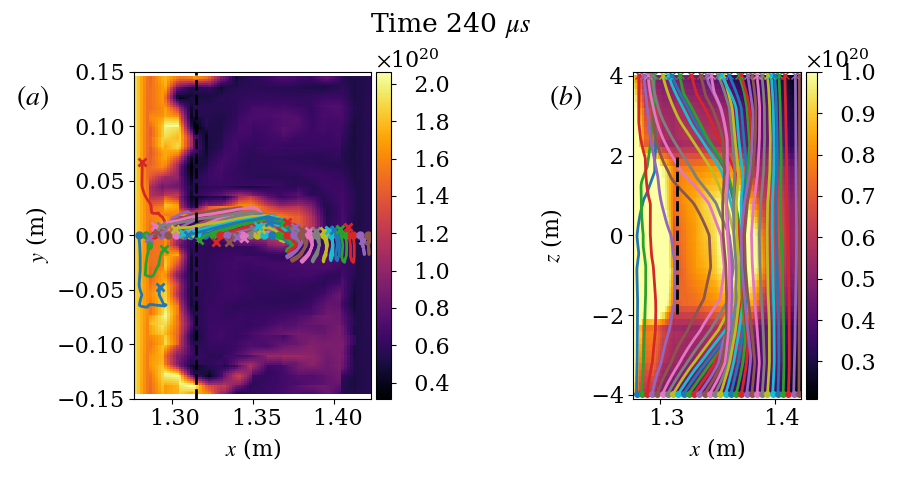}
    \caption{Three-dimensional magnetic field line trajectories at $t=240\ \mu$s, projected onto the $x-y$ plane in $(a)$ and the $x-z$ plane in $(b)$. The ion density is plotted in the background, at $z=0$ m in $(a)$ and averaged over $|y|<0.02$ m in $(b)$. Each field line starts at $y=0$ m and $z=-4$ m for a range of $x$ values and is traced to $z=4$ m. The starting points are marked with circles and the ending points are marked with crosses. Each field line is colored the same in both $(a)$ and $(b)$. The field lines in the near-SOL are stretched radially outward by a blob near $y=0$ m. }
    \label{fig:bstream-x}
\end{figure}

In Figures \ref{fig:bstream-y} and \ref{fig:bstream-x} we show projections of the three-dimensional magnetic field line trajectories. These plots are created by integrating the field line equations for the total (background plus fluctuation) magnetic field. In Figure \ref{fig:bstream-y}, each field line starts at $z=-4$ m and either $x=1.33$ m or $x=1.38$ m for a range of $y$ values and is traced to $z=4$ m. The starting points (at $z=-4$ m ) are marked with circles, while the ending points (at $z=4$ m) are marked with crosses. The trajectories have been projected onto the $x-y$ plane, and we have also plotted the ion density at $z=0$ m in the background. From left to right, we show a short time series of snapshots, with $t=230,\ 240$ and $250\ \mu$s. At $t=230\ \mu$s, a blob is starting to emerge from the source region at $y\approx0.04$ m. The field lines that start at $x=1.33$ m are beginning to be stretched radially outward as the blob emerges. In the $t=240\ \mu$s snapshot, we see that the blob is now propagating radially outward into the SOL region and the $x=1.33$ m field lines have been stretched further. The field lines that start at $x=1.38$ m are now also starting to be stretched near $y\approx0.2$ m, and they are stretched even more in the $t=250\ \mu$s snapshot as the blob continues to propagate. We can also see the remnants of another blob that was ejected near $y=-0.1$ m in previous frames. In the $t=230\ \mu$s snapshot, the field lines have been stretched by this blob, but by $t=250\ \mu$s the field lines in this region have returned closer to their equilibrium position. This behaviour of blobs bending and stretching the field lines is an inherently full-$f$ phenomenon. The blobs have a higher density and temperature than the background, so they raise the local plasma $\beta$ as they propagate. This causes the field lines to move with the plasma, allowing the fields lines to be deformed and stretched by the radially propagating blobs and ultimately leading to larger magnetic fluctuations.

In Figure \ref{fig:bstream-x} we show a slightly different view of the field-line trajectories at $t=240\ \mu$s. Field lines are still traced from the bottom ($z=-4$ m) to the top ($z=4$ m), but now each field line starts at $y=0$ m for a range of $x$. The starting points are again marked with circles and the ending points are marked with crosses. We have projected the three-dimensional trajectories onto the $x-y$ plane in Figure \ref{fig:bstream-x}$(a)$, and onto the $x-z$ plane in Figure \ref{fig:bstream-x}$(b)$. In $(a)$ we again plot the ion density at $z=0$ m in the background; in $(b)$ the ion density has been averaged over $|y|<0.02$ m. As can be seen in Figure \ref{fig:bstream-x}$(b)$, the blob propagating near $y\approx0$ m has stretched several field lines radially outward near the midplane. These bowed-out field lines originate from a range of $x$ values, $1.3\ $m $\lesssim x \lesssim 1.35$ m, and have all been dragged along with the blob as it was ejected from the source region and propagated radially outward. We also see some degree of line tying in these plots, with many of the field lines ending at a similar point in $x-y$ space to where they began, despite being stretched near the midplane. {The field lines are not perfectly line tied, however; if they were, the crosses would perfectly align with their corresponding circles in the $x-y$ projections. Because our sheath boundary condition allows current fluctuations at the sheath interface, we can model the finite resistance of the sheath, which makes line tying only partial \citep{kunkel1966interchange}. This allows the footpoints of the field lines to slip at the sheath interface \citep{ryutov2006dynamics}. Examining figures \ref{fig:bstream-y} and \ref{fig:bstream-x}, we see evidence of this in the simulation, with most of the end points moving slowly and smoothly in the vicinity of their origin, especially at larger $x$. In the source region, however, there are other field lines whose end points suddenly jump further away from their origin. This suggests that we are seeing line breaking 
(reconnection) due to electron inertia effects and numerical diffusion, as field lines are pushed close together by large perturbations in the source region.
}

\begin{figure}
    \centering
    \includegraphics[width=\textwidth]{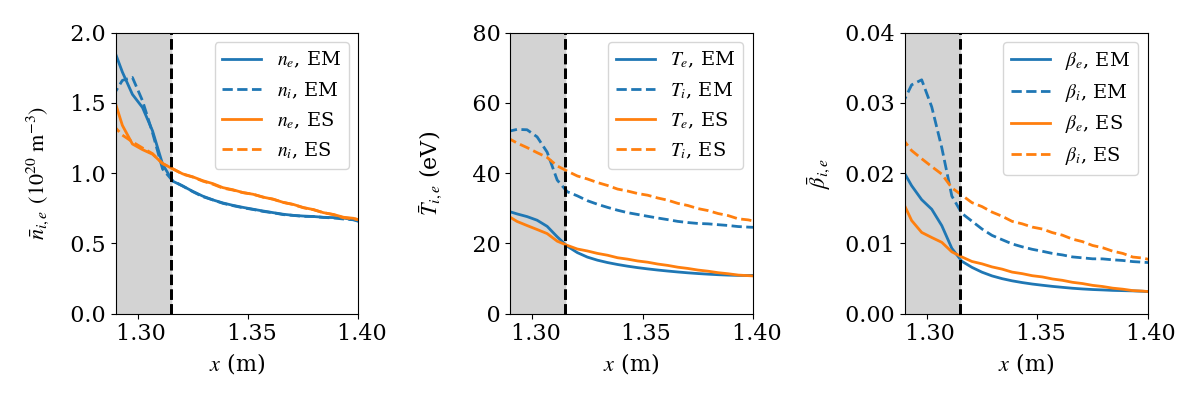}
    \caption{Radial profiles of density (left), temperature (middle) and beta (right) for electrons (solid) and ions (dashed). Profiles from the electromagnetic case (EM) are blue, and the electrostatic profiles (ES) are yellow. The profiles are averaged in $y$, $z$ and time using data near the midplane $(|z|<0.4\ \mathrm{ m})$ over a period of 400 $\mu$s. The electromagnetic case shows shallower profiles in the SOL region, indicating that there is less radial particle and heat transport (as confirmed by Figure \ref{fig:partflux}).}
    \label{fig:profiles}

    \centering
    \includegraphics[width=.65\textwidth]{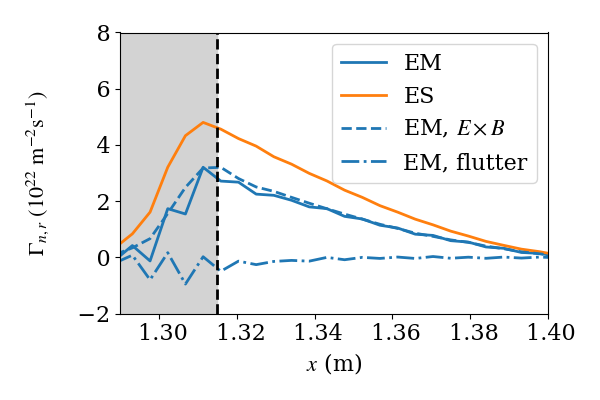}
    \caption{Radial profile of the radial electron particle flux $\Gamma_{n,r}$, averaged in $y$, $z$ and time using data near the midplane $(|z|<0.4\ \mathrm{ m})$ over a period of 400 $\mu$s. The transport in the electromagnetic case (EM, blue) is roughly 40\% lower than in the electrostatic case (ES, yellow). This contributes to the shallower electron density profile in the electromagnetic case, as seen in Figure \ref{fig:profiles}. The electromagnetic case contains radial transport from both $E\times B$ drift (dashed) and magnetic flutter (dot-dashed).}
    \label{fig:partflux}
\end{figure} 

\begin{figure}
    \centering
    \includegraphics[width=\textwidth]{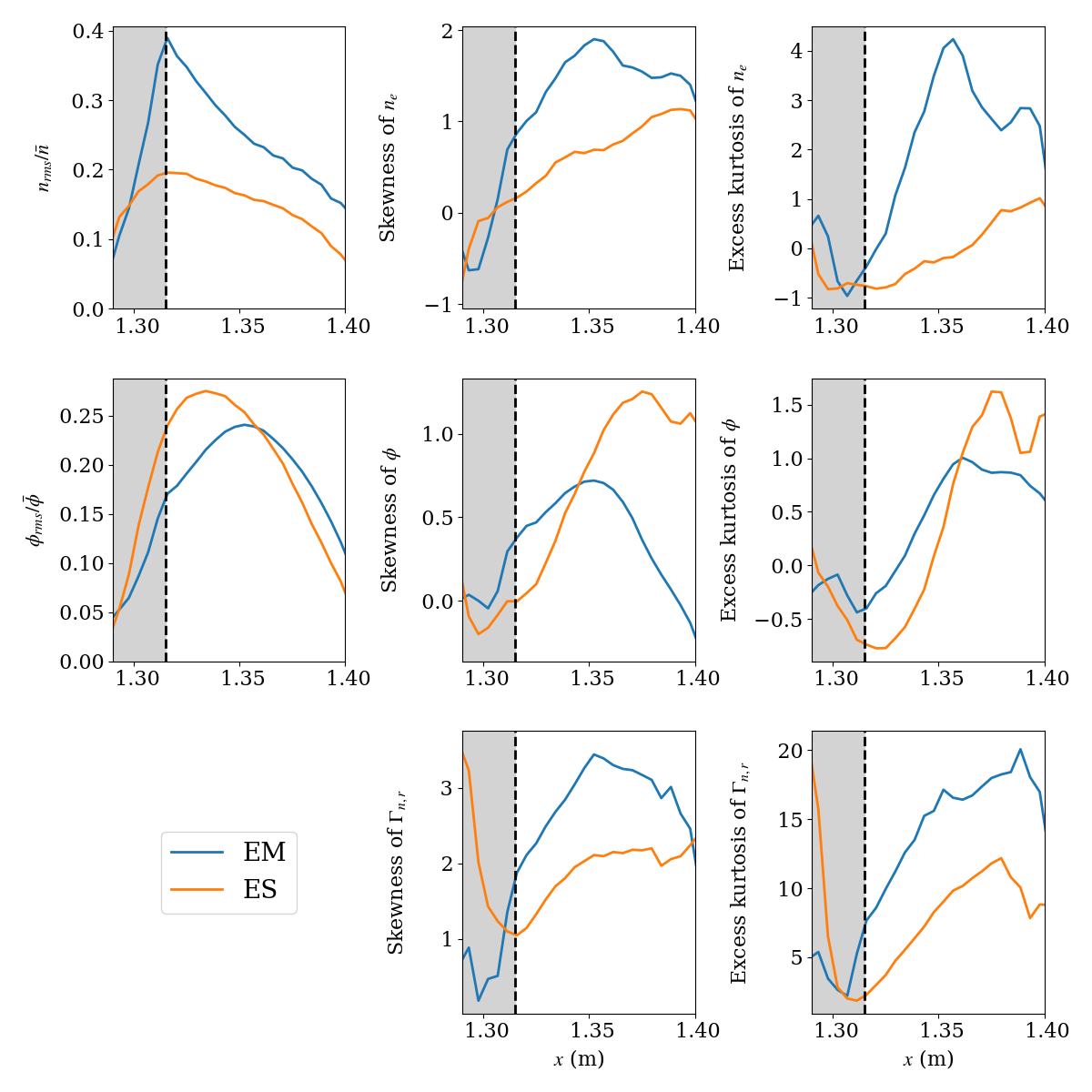}
    \caption{Comparison of fluctuation statistics for the electron density (top row), electrostatic potential {(middle row), and radial electron particle flux (bottom row)} between the electromagnetic case (EM, blue) and a corresponding electrostatic case (ES, yellow). From left to right, we show radial profiles of the normalized RMS fluctuation amplitude, skewness, and excess kurtosis. These plots were computed by averaging in $y$, $z$ and time using data near the midplane $(|z|<0.4\ \mathrm{ m})$ over a period of 400 $\mu$s. The electromagnetic case shows higher electron density fluctuation amplitude, skewness, and excess kurtosis. This is an indication that the electromagnetic case has more intermittent density fluctuations. {The skewness and kurtosis of the particle flux also indicate that the transport is more intermittent in the electromagnetic case.} The statistics for the potential are relatively similar for the electrostatic and electromagnetic cases.}
    \label{fig:ne_stats}
\end{figure}

We have also run a corresponding electrostatic simulation in this configuration for  direct comparison. This simulation is identical in configuration to the $L_z=8$ m case from \citet{shi2019full} except for the increased particle source rate and lack of cross-species collisions. In Figure \ref{fig:profiles} we show a comparison of radial profiles of density, temperature and $\beta$ for the electromagnetic and electrostatic cases. The profiles for the electromagnetic case are shallower in the SOL region and steeper in the source region. This suggests that there is less radial particle and heat transport into the SOL region in the electromagnetic case. This is in part confirmed by the profile of the radial particle flux in Figure \ref{fig:partflux}, showing approximately 40\% less particle transport in the electromagnetic case. The total particle flux $\Gamma_{n,r}$ includes the $E\times B$ particle flux, $\Gamma_{n,r,E\times B}=\langle \tilde{n}_e \tilde{v}_r \rangle$, with $v_r = E_r/B = -({1}/{B})\partial{\phi}/\partial{y}$. In the electromagnetic case, $\Gamma_{n,r}$ also includes the particle flux due to magnetic flutter, $\Gamma_{n,r,\text{flutter}}=\langle \widetilde{n_e u_{\parallel e}} b_r \rangle$, with $b_r =({1}/{B})\partial{A_\parallel}/\partial{y}$. The tilde indicates the fluctuation of a time-varying quantity, defined as $\tilde{A}=A-\bar{A}$ with $\bar{A}$ the time average of $A$. The brackets $\langle A \rangle$ denote an average in $y$, $z$ (near the midplane), and time. We also note that the electromagnetic profiles might be even shallower in the SOL region if not for the floor on the source used to prevent positivity issues in the distribution function at large $x$.

We also compare fluctuation statistics between the electromagnetic and electrostatic cases in Figure \ref{fig:ne_stats}. Statistics of the electron density are shown on the top row, the {middle row} shows statistics of electrostatic potential fluctuations{ and the bottom row shows statistics of the radial electron particle flux.} Despite the fact that the electromagnetic case shows less particle transport, the root-mean-square relative density fluctuations are larger in the electromagnetic case by up to a factor of two. The electromagnetic case also has higher skewness and excess kurtosis, indicating that the density fluctuations in the electromagnetic case are more intermittent. {The skewness and kurtosis of the particle flux also indicate more intermittency in the electromagnetic case. This contributes to the reduced transport shown in the electromagnetic case, since the transport events are rarer even if the fluctuation levels are larger.} Meanwhile, the fluctuation statistics for the potential are relatively similar between the electromagnetic and electrostatic cases. The statistics of the density and potential appear more coupled in the electrostatic case, consistent with the electrons behaving adiabatically when electromagnetic effects are neglected. In the electromagnetic case the density fluctuations are more intermittent and higher amplitude than the potential fluctuations.

Finally, we note that in terms of computational cost, the electromagnetic simulation is less than twice as expensive as the corresponding electrostatic simulation on a per-time-step basis. On 128 cores, the time per time step was 0.41 s/step for the electrostatic simulation and 0.68 s/step for the electromagnetic simulation. The increased cost is due to the additional field solves required for Ohm's law, along with additional terms in the gyrokinetic equation. However, due to time-step restrictions on an electrostatic simulation due to the electrostatic shear Alfv\'en mode (also known as the $\omega_H$ mode) \citep{lee1987gyrokinetic}, the electromagnetic simulation makes up some of the additional cost by taking slightly larger time steps. The total wall-clock time (on 128 cores) for the electrostatic simulation was approximately 65 h, and the electromagnetic simulation took about 82 h. Altogether, the cost of these simulations is relatively modest, and the addition of electromagnetic effects only makes the simulations marginally ($\sim 25$\%) more expensive. We also note that the new version of \gke, which uses a quadrature-free modal DG scheme, is approximately 10 times faster than the previous version of \gke used in \citet{shi2019full}, which used nodal DG with Gaussian quadrature. More details about the improvements from the quadrature-free modal scheme will be reported elsewhere.

\section{Summary \& Conclusion}
\label{sec:conclusion}

In this paper we have presented an energy-conserving scheme for the full-$f$ electromagnetic gyrokinetic system. We choose the symplectic formulation of EMGK, which uses the parallel velocity as an independent variable. This leads to the time derivative of the parallel vector potential, $\pderivInline{A_\parallel}{t}$, appearing explicitly in the gyrokinetic equation. We handle this term directly by solving an Ohm's law. We presented the conservation properties of the EMGK system.

We described the discontinuous Galerkin scheme used to discretize the EMGK system in phase space. We proved that the scheme preserves particle conservation, and that the scheme also preserves energy conservation provided that the discrete Hamiltonian is continuous. This is achieved by using the (continuous) finite-element method for the field solves. We also detailed a basic forward-Euler time-stepping scheme to be used in the stages of a multi-stage high-order SSP-RK scheme. The time-stepping scheme updates the gyrokinetic equation in two parts, with the result of the first part [denoted  $\partial{(\mathcal{J}f_s)^\star}/\partial{t}$] being used in Ohm's law to solve directly for $\pderivInline{A_\parallel}{t}$ so that it can then be used in the second part of the gyrokinetic update.

We have implemented the scheme in the gyrokinetic solver of the \gke computational plasma framework. We provided two linear benchmarks to validate the electromagnetic scheme: a kinetic Alfv\'en wave calculation and a local kinetic ballooning mode instability calculation. In both cases results from \gke agree well with analytic results. The success of these calculations, especially in cases with $\hat{\beta}/k_\perp^2\rho_s^2\gg1$, indicates that the scheme avoids the Amp\`ere cancellation problem. {Further, the discontinuous Galerkin nature of the scheme enables (but does not enforce) exact cancellation of the components of $E_\parallel$, allowing the system to capture the MHD limit with $E_\parallel=0$.}

We presented a nonlinear electromagnetic full-$f$ gyrokinetic simulation of turbulence on helical, open field lines as a rough model of the tokamak scrape-off layer. This simulation is the first nonlinear electromagnetic gyrokinetic simulation on open field lines. We showed data illustrating the interplay between blobs propagating into the SOL and the resulting bending and stretching of magnetic field lines. We also made quantitative comparisons between the electromagnetic simulation and a corresponding electrostatic simulation. Notably, the electromagnetic simulation exhibits less transport and shallower density and temperature profiles in the SOL, as well as larger root-mean-square density fluctuations with more intermittency. {Future work will examine the particle and energy balance in these nonlinear simulations to confirm the conservation properties proved in Section \ref{discons} after accounting for sources and losses to the walls.}

A number of extensions have been left to future work. The capability to simulate more realistic magnetic geometries using a non-orthogonal field-aligned coordinate system is currently in progress, so that effects of magnetic shear and non-constant curvature can be included. The inclusion of both open and closed-field-line regions, including the X-point in diverted geometries, is  an additional complication that must be addressed. Gyroaveraging is another important effect that must be implemented to improve fidelity. Further, a robust solution to the issue of maintaining positivity of the distribution function has been implemented for Hamiltonian terms and is in progress for the collision operator. This could, for example, alleviate the need to use a source floor in the nonlinear simulations presented in Section \ref{sec:nonlinear}, which could further enhance the differences between the electromagnetic and electrostatic profiles.

The modest cost of the nonlinear full-$f$ gyrokinetic simulations that we have presented make the prospect of using \gke for whole-device modelling a feasible goal. The inclusion of electromagnetic effects will be crucial to the fidelity of these efforts. Such a tool would be invaluable to future studies of turbulent transport in fusion devices, both from a theoretical perspective and also as a model for predicting and analysing the performance of current and future experiments.

\section*{Acknowledgements}
We would like to thank Tess Bernard, Petr Cagas, James Juno and other members of the \gke team for helpful discussions and support, including the development of the \texttt{postgkyl} post-processing tool which facilitated the creation of many figures in this paper.
 Research support came from the U.S. Department of Energy: NRM was
supported by the DOE CSGF program, provided under grant 
DE-FG02-97ER25308;
 AH and GWH were supported by the Partnership for Multiscale Gyrokinetic Turbulence (MGK) Project, part of the DOE Scientific Discovery Through Advanced Computing (SciDAC) program, via DOE contract DE-AC02-09CH11466 for the Princeton Plasma Physics Laboratory; MF was supported by DOE contract DE-FC02-08ER54966. Computations were performed on the Eddy cluster at the Princeton Plasma Physics Laboratory.

\section*{Supplementary material}  
Supplementary material is available at https://doi.org/10.1017/S0022377820000070.

\appendix

\section{Amp\`ere cancellation problem}
\label{app:cancel}
To understand the root of the Amp\`ere cancellation problem, we examine the simple Alfv\'en wave case from Section \ref{res:kaw}. Recall that in this simple case, the gyrokinetic system is given by
\begin{gather}
    \pderiv{f_e}{t} = \{H_e,f_e\}-\frac{e}{m_e}\pderiv{f_e}{v_\parallel}\pderiv{A_\parallel}{t} = - v_\parallel \pderiv{f_e}{z} {-} \frac{e}{m}\pderiv{f_e}{v_\parallel}\left(\pderiv{\phi}{z}+\pderiv{A_\parallel}{t} \right) \\
    k_\perp^2 \frac{m_i n_{0}}{B^2}\phi = e n_0-e\int dv_\parallel\ f_e \label{app:alf-poisson} \\
    \left(k_\perp^2 + C_N\ \frac{\mu_0 e^2}{m_e}\int dv_\parallel f_e\right)\pderiv{A_{\parallel}}{t} = -C_{J}\ \mu_0 e\int dv_\parallel\ v_\parallel \{H_e, f_e\} \label{app:alf-ohm}.
\end{gather}
Note that now we have inserted two constants, $C_N$ and $C_{J}$, in the integrals in Eq. (\ref{app:alf-ohm}). We will use these constants to represent small errors that could arise in the numerical calculation of these integrals. As in Section \ref{res:kaw}, we can calculate the dispersion relation for this system, but now we will take the limit $\omega\gg k_\parallel v_{te}$, so that the dispersion relation reduces to
\begin{equation}
    \omega^2 = \frac{k_\parallel^2 v_A^2}{C_N + k_\perp^2 \rho_s^2/\hat{\beta}}\left[1+(C_N - C_{J})\frac{\hat{\beta}}{k_\perp^2\rho_s^2}\right], \label{cancel-disp}
\end{equation}
where recall that $\hat{\beta}=(\beta_e/2) m_i/m_e$. This reduces to the correct result if $C_N=C_{J}=1$. However, if $C_N\neq C_{J}$, there will be large errors for modes with $\hat{\beta}/(k_\perp^2\rho_s^2)\gg 1$. This means the two integrals in Eq. (\ref{app:alf-ohm}) must be calculated carefully and consistently so that any errors exactly cancel.

\subsection{Hamiltonian ($p_\parallel$) formulation}
We now briefly discuss the cancellation problem in the Hamiltonian ($p_\parallel$) formulation of gyrokinetics. In this case, the simple system above becomes
\begin{gather}
    \pderiv{f_e}{t} = \{H_e,f_e\}= - \frac{1}{m_e}p_\parallel \pderiv{f_e}{z} {-} e\pderiv{f_e}{p_\parallel}\pderiv{\phi}{z} \\
    k_\perp^2 \frac{m_i n_{0}}{B^2}\phi = e n_0-e\int dp_\parallel\ f_e \label{app:alf-poissonH} \\
    \left(k_\perp^2 + C_N\ \frac{\mu_0 e^2}{m_e^2}\int dp_\parallel\  f_e\right)A_{\parallel} = -C_{J}\ \frac{\mu_0 e}{m_e^2}\int dp_\parallel\ p_\parallel f_e \label{app:alf-ohmH},
\end{gather}
where we have again included constants $C_N$ and $C_{J}$ to represent numerical errors in the integrals. The resulting dispersion relation is identical to Eq. (\ref{cancel-disp}), so the two integrals must again be calculated carefully and consistently so that errors exactly cancel. This could be slightly more challenging than in the symplectic case. Suppose the $p_\parallel$ grid does not extend to infinity but has some finite limits which are expected to be large enough in practice. These limits are used when numerically computing the integrals. Since $p_\parallel$ depends on a time-dependent quantity in $A_\parallel$, it is possible that the $p_\parallel$ limits may need to be time-dependent if $A_\parallel$ fluctuations are large in order to consistently compute the integrals.

\section{The discrete weak form of Ohm's law} \label{app:Ohm}
To obtain the discrete weak form of Ohm's law, we start by taking the time derivative of Eq. (\ref{FEMampere-local}):
\begin{equation}
\int\limits_{\mathcal{K}^R_j} \!\!d\v{R}\,\nabla_\perp \pderiv{A_{\parallel h}}{t} \v{\cdot} \nabla_\perp \varphi^{(j)} - \oint\limits_{\partial \mathcal{K}^R_j} \!\!d\v{s}_R\v{\cdot}\nabla_\perp \pderiv{A_{\parallel h}}{t}\ \varphi^{(j)}\  = \mu_0\sum_s\frac{q_s}{m_s} \int\limits_{\mathcal{K}^R_j} \!\!d\v{R}\, \varphi^{(j)} \!\!\int\limits_{\mathcal{T}^v}\!\!d\v{w}\, \pderiv{H_{s\,h}}{v_\parallel} \pderiv{(\mathcal{J}f_{s\,h})}{t}. \label{ohm1}
\end{equation}
Now, note that, analogously to Eq. (\ref{fstar}), we can write the discrete weak form of the gyrokinetic equation as
\begin{align}
    &\int\limits_{\mathcal{K}_j}\!\!d\v{R}\, d\v{w}\, \psi \pderiv{(\mathcal{J}f_h)}{t} =\notag\\ &\quad\int\limits_{\mathcal{K}_j}\!\!d\v{R}\, d\v{w}\, \psi \pderiv{(\mathcal{J}f_h)}{t}^\star - \oint\limits_{\partial \mathcal{K}_j}\!\!d\v{R}\, ds_w\, \!\!\left(\dot{v}^H_{\parallel h}-\frac{q}{m}\pderiv{A_{\parallel h}}{t}\right) \psi^- \widehat{\mathcal{J}f_h} - \int\limits_{\mathcal{K}_j}\!\!d\v{R}\, d\v{w}\, \mathcal{J}f_h \frac{q}{m}\pderiv{A_{\parallel h}}{t} \pderiv{\psi}{v_\parallel}. \label{app:DGstar}
\end{align}
where
\begin{align}
    \int\limits_{\mathcal{K}_j}d\v{R}\ d\v{w}\ \psi \pderiv{(\mathcal{J}f_h)^\star}{t} = &-\oint\limits_{\partial \mathcal{K}_j} d\v{w}\ d\v{s}_R\v{\cdot}\dot{\v{R}}_h \psi^- \widehat{\mathcal{J}f_h}+\int\limits_{\mathcal{K}_j}d\v{R}\ d\v{w}\ \mathcal{J}f_h \dot{\v{R}}_h\v{\cdot} \nabla \psi \notag \\&\quad+ \int\limits_{\mathcal{K}_j}d\v{R}\ d\v{w}\ \mathcal{J}f_h \dot{v}^H_{\parallel h} \pderiv{\psi}{v_\parallel} + \int\limits_{\mathcal{K}_j}d\v{R}\ d\v{w}\ \psi\left(\mathcal{J}C[f_h] + \mathcal{J}S\right).
\end{align}
Substituting $\psi=\varphi^{(j)} \pderivInline{H_h}{v_\parallel}$ in Eq. (\ref{app:DGstar}) and summing over velocity cells, we obtain
\begin{align}
    &\int\limits_{\mathcal{K}^R_j}\!\! d\v{R}\, \varphi^{(j)}\!\! \int\limits_{\mathcal{T}^v}\!\!d\v{w}\, \pderiv{H_{h}}{v_\parallel} \pderiv{(\mathcal{J}f_{h})}{t} = \!\int\limits_{\mathcal{K}^R_j}\!\! d\v{R}\, \varphi^{(j)}\!\! \int\limits_{\mathcal{T}^v}\!\!d\v{w}\, \pderiv{H_{h}}{v_\parallel} \pderiv{(\mathcal{J}f_{h})^\star}{t} \notag\\
    &\qquad- \!\int\limits_{\mathcal{K}_j^R}\!\! d\v{R}\, \varphi^{(j)}\sum\limits_i\!\oint\limits_{\partial\mathcal{K}_i^v}\!\!ds_w\! \left(\dot{v}^H_{\parallel h}-\frac{q}{m}\pderiv{A_{\parallel h}}{t}\right) \pderiv{H_h}{v_\parallel}^- \widehat{\mathcal{J}f_h} \notag \\
    &\qquad- \int\limits_{\mathcal{K}_j^R}d\v{R}\ \varphi^{(j)} \frac{q}{m}\pderiv{A_{\parallel h}}{t} \int\limits_{\mathcal{T}^v} d\v{w}\ \mathcal{J}\pderiv{^2H_h}{v_\parallel^2}f_h.
\end{align}
Note that, for $p_v>1$, the $v_\parallel$ surface term on the right-hand side vanishes because $\pderivInline{H_h}{v_\parallel}$ is continuous across $v_\parallel$ cell interfaces when $v_\parallel^2$ is included in the basis, resulting in cancellations. However, for $p_v=1$ this term is not continuous, and we must keep this surface term; further, the {last} term on the right-hand side vanishes for $p_v=1$ since $\pderivInline{^2H_h}{v_\parallel^2}=0$.
We can now substitute this result into the right-hand side of Eq. (\ref{ohm1}), giving
\begin{align}
    \int\limits_{\mathcal{K}^R_j}\!\!d\v{R}\, &\nabla_\perp \pderiv{A_{\parallel h}}{t} \v{\cdot} \nabla_\perp \varphi^{(j)} - \!\!\oint\limits_{\partial \mathcal{K}^R_j} \!\!\!d\v{s}_R\v{\cdot}\nabla_\perp \pderiv{A_{\parallel h}}{t}\ \varphi^{(j)}
    - \!\!\int\limits_{\mathcal{K}_j^R}\!\!d\v{R}\, \varphi^{(j)}\pderiv{A_{\parallel h}}{t} \sum_{s,i} \frac{\mu_0 q_s^2}{m_s}\!\oint\limits_{\partial\mathcal{K}^v_i} \!\!ds_w\, \bar{v}_\parallel^- \widehat{\mathcal{J} f_{s\,h}} \notag \\
    &= \mu_0\sum_s q_s \!\!\int\limits_{\mathcal{K}^R_j} \!\!d\v{R}\, \varphi^{(j)} \left[\int\limits_{\mathcal{T}^v}\!\!d\v{w}\, \bar{v}_\parallel \pderiv{(\mathcal{J}f_{s\,h})}{t}^\star-\sum_i\!\oint\limits_{\partial \mathcal{K}_i^v}\!\!ds_w\,  \bar{v}_\parallel^-{\dot{v}^H_{\parallel h}} \widehat{\mathcal{J}f_{s\,h}} \right], \qquad(p_v=1) \label{app:Ohmp1} \\
   \int\limits_{\mathcal{K}^R_j}\!\!d\v{R}\, &\nabla_\perp \pderiv{A_{\parallel h}}{t} \v{\cdot} \nabla_\perp \varphi^{(j)} - \!\!\oint\limits_{\partial \mathcal{K}^R_j} \!\!\!d\v{s}_R\v{\cdot}\nabla_\perp \pderiv{A_{\parallel h}}{t}\ \varphi^{(j)}
    + \!\!\int\limits_{\mathcal{K}_j^R}\!\!d\v{R}\, \varphi^{(j)}\pderiv{A_{\parallel h}}{t} \sum_s \frac{\mu_0 q_s^2}{m_s}\!\int\limits_{\mathcal{T}^v}\!\! d\v{w}\, \mathcal{J} f_{s\,h} \notag \\
    &=\mu_0\sum_s q_s \!\!\int\limits_{\mathcal{K}^R_j} \!\!d\v{R}\, \varphi^{(j)}\!\int\limits_{\mathcal{T}^v}\!\!d\v{w}\, v_\parallel \pderiv{(\mathcal{J} f_{s\,h})}{t}^\star, \qquad (p_v>1) \label{app:Ohmp2}
\end{align}
In Eq. (\ref{app:Ohmp1}), $\bar{v}_\parallel$ is the piecewise-constant projection of $v_\parallel$. 

\section{{Semi-discrete dispersion relation for Alfv\'en wave}}\label{app:disp}
In this appendix we derive a semi-discrete Alfv\'en wave dispersion relation using a piecewise-linear DG discretization for only the $v_\parallel$ coordinate. 
The main purpose of this appendix is to show how our discrete scheme avoids the Amp\`ere cancellation problem. We will also show how the integrals in the DG weak form are computed analytically in our modal scheme.

The semi-discrete gyrokinetic weak form for this system is
\begin{gather}
    \int\limits_{\mathcal{K}_j} dv_\parallel\ \psi \pderiv{f_h}{t} + \int\limits_{\mathcal{K}_j}dv_\parallel\ \psi \frac{1}{m_e}\pderiv{H_h}{v_\parallel}\pderiv{f_h}{z} - \frac{e}{m_e}\left(\pderiv{\phi}{z}+\pderiv{A_\parallel}{t}\right)\int\limits_{\mathcal{K}_j}dv_\parallel\ \pderiv{\psi}{v_\parallel}f_h \qquad\notag \\
    + \frac{e}{m_e}\left(\pderiv{\phi}{z}+\pderiv{A_\parallel}{t}\right)(\psi^- \widehat{f}_h)\bigg\rvert_{\partial\mathcal{K}_j} = 0.
\end{gather}
We begin by mapping each cell $\mathcal{K}_j$ to $\xi\in [-1,1]$ via the transformation $\xi = 2(v_\parallel - \bar{v}_\parallel^j)/\Delta v_\parallel$, where $\bar{v}_\parallel^j$ is the cell centre of cell $j$
\begin{gather}
    \int\limits_{-1}^1 d\xi\ \psi \pderiv{f_h}{t} + \int\limits_{-1}^1 d\xi\ \psi \frac{2}{\Delta v_\parallel}\frac{1}{m_e}\pderiv{H_h}{\xi}\pderiv{f_h}{z} - \frac{e}{m_e}\left(\pderiv{\phi}{z}+\pderiv{A_\parallel}{t}\right)\int\limits_{-1}^1 d\xi \frac{2}{\Delta v_\parallel}\pderiv{\psi}{\xi}f_h \qquad \notag \\
    + \frac{e}{m_e}\left(\pderiv{\phi}{z}+\pderiv{A_\parallel}{t}\right)\frac{2}{\Delta v_\parallel}(\psi^- \widehat{f}_h)\bigg\rvert_{-1}^{\ 1} = 0.
\end{gather}
Taking an orthonormal piecewise-linear basis in $\xi$, $\psi = [\frac{1}{\sqrt{2}}, \frac{\sqrt{3}}{\sqrt{2}}\xi]$, we expand $f_h$ on the basis in cell $j$ as
\begin{equation}
    f_h^j(z,v_\parallel,t) = \sum_\ell \psi_\ell(\xi) f_\ell^j(z,t) = \frac{1}{\sqrt{2}}f_0^j + \frac{\sqrt{3}}{\sqrt{2}} f_1^j \xi.
\end{equation}
(Note that in the fully discretized case all coordinate dependence would be contained in multi-variate basis functions.) We can then analytically integrate the weak form for each $\psi_\ell$ to obtain the modal evolution equation for each DG `mode' $f_\ell$:
\begin{gather}
    \pderiv{f_0^j}{t} + \bar{v}_\parallel^j\pderiv{f_0^j}{z} + \frac{e}{m_e}\left(\pderiv{\phi}{z} + \pderiv{A_\parallel}{t}\right)\frac{\sqrt{2}}{\Delta v_\parallel}\widehat{f}_{h}^{\,j}\bigg\rvert_{-1}^{\ 1} = 0 \\
    \pderiv{f_1^j}{t} + \bar{v}_\parallel^j\pderiv{f_1^j}{z} + \frac{e}{m_e}\left(\pderiv{\phi}{z} + \pderiv{A_\parallel}{t}\right)\frac{\sqrt{6}}{\Delta v_\parallel}\xi\widehat{f}_{h}^{\,j}\bigg\rvert_{-1}^{\ 1} - \frac{e}{m_e}\left(\pderiv{\phi}{z} + \pderiv{A_\parallel}{t}\right)\frac{2\sqrt{3}}{\Delta v_\parallel}f_0^j = 0.
\end{gather}
Finally, we will make the ansatz $f_\ell = F_{M\ell} + f_\ell e^{i k_\parallel z- i \omega t}$ and linearize:
\begin{gather}
    -i(\omega-k_\parallel \bar{v}_\parallel^j){f_0^j} + \frac{e}{m_e}\left(ik_\parallel{\phi} + \pderiv{A_\parallel}{t}\right)\frac{\sqrt{2}}{\Delta v_\parallel}\widehat{F}_{Mh}^{j}\bigg\rvert_{-1}^{\ 1} = 0 \label{f0w} \\
    -i(\omega-k_\parallel \bar{v}_\parallel^j){f_1^j} + \frac{e}{m_e}\left(ik_\parallel{\phi} + \pderiv{A_\parallel}{t}\right)\frac{\sqrt{6}}{\Delta v_\parallel}\xi\widehat{F}_{Mh}^{j}\bigg\rvert_{-1}^{\ 1} - \frac{e}{m_e}\left(ik_\parallel{\phi} + \pderiv{A_\parallel}{t}\right)\frac{2\sqrt{3}}{\Delta v_\parallel}F_{M0}^j = 0.
\end{gather}

We now turn to the field equations. The Poisson equation is
\begin{equation}
    k_\perp^2\frac{m_i n_0}{B^2}\phi = e n_0 - e\sum_j \int\limits_{\mathcal{K}_j} dv_\parallel\ f_h.
\end{equation}
Expanding $f_h$ and using the ansatz, this becomes
\begin{equation}
    k_\perp^2\frac{m_i n_0}{B^2}\phi = e n_0 - e\sum_j \frac{\Delta v_\parallel}{\sqrt{2}}F_{M0}^j - e \sum_j \frac{\Delta v_\parallel}{\sqrt{2}} f_0^j = - e \sum_j \frac{\Delta v_\parallel}{\sqrt{2}} f_0^j,
\end{equation}
where we will define $F_{Mh}$ so that $\sum_j \frac{\Delta v_\parallel}{\sqrt{2}}F_{M0}^j = n_0$ by definition. For Ohm's law, we must use the $p_v=1$ form from Eq. (\ref{Ohmp1}), which gives
\begin{equation}
    k_\perp^2\pderiv{A_\parallel}{t} - \pderiv{A_\parallel}{t} \frac{\mu_0 e^2}{m_e}\sum_j \bar{v}_\parallel^j \widehat{f}_{h}^{\,j}\bigg\rvert_{\partial \mathcal{K}_j} = - \mu_0 e \sum_j \int\limits_{\mathcal{K}_j}dv_\parallel\ \bar{v}_\parallel^j \pderiv{f_h}{t}^\star + ik_\parallel \phi \frac{\mu_0 e^2}{m_e}\sum_j \bar{v}_\parallel^j \widehat{f}_{h}^{\,j}\bigg\rvert_{\partial \mathcal{K}_j},
\end{equation}
where
\begin{equation}
    \int\limits_{\mathcal{K}_j}dv_\parallel\ \psi \pderiv{f_h}{t}^\star = -i k_\parallel \int\limits_{\mathcal{K}_j}dv_\parallel\ \psi \frac{1}{m_e}\pderiv{H_h}{v_\parallel} f_h + \frac{e}{m_e}ik_\parallel\phi\int\limits_{\mathcal{K}_j}dv_\parallel\ \pderiv{\psi}{v_\parallel}f_h.
\end{equation}
Again expanding and using the ansatz, Ohm's law becomes
\begin{equation}
k_\perp^2\pderiv{A_\parallel}{t} - \pderiv{A_\parallel}{t} \frac{\mu_0 e^2}{m_e}\sum_j \bar{v}_\parallel^j \widehat{F}_{Mh}^{j}\bigg\rvert_{-1}^{\ 1} = - \mu_0 e(i k_\parallel) \sum_j \frac{\Delta v_\parallel}{\sqrt{2}}\bar{v}_\parallel^{j\,2} f_0^j + ik_\parallel \phi \frac{\mu_0 e^2}{m_e}\sum_j \bar{v}_\parallel^j \widehat{F}_{Mh}^{j}\bigg\rvert_{-1}^{\ 1}. \label{ohmw} 
\end{equation}
Analogously to Appendix \ref{app:cancel}, we can rewrite this equation as
\begin{equation}
k_\perp^2\pderiv{A_\parallel}{t} + \pderiv{A_\parallel}{t} \frac{\mu_0 e^2 n_0}{m_e} C_N = - \mu_0 e(i k_\parallel) \sum_j \frac{\Delta v_\parallel}{\sqrt{2}}\bar{v}_\parallel^{j\,2} f_0^j - ik_\parallel \phi \frac{\mu_0 e^2 n_0}{m_e}C_J, \label{ohmw2} 
\end{equation}
where we have defined
\begin{gather}
    C_N = -\sum_j\frac{1}{n_0}\bar{v}_\parallel^j \widehat{F}_{Mh}^{j}\bigg\rvert_{-1}^{\ 1}, \\
    C_J = -\sum_j\frac{1}{n_0}\bar{v}_\parallel^j \widehat{F}_{Mh}^{j}\bigg\rvert_{-1}^{\ 1}.
\end{gather}
Clearly $C_N = C_J$, which allows us to move the $C_N$ term to the right-hand side, giving a term proportional to the total parallel electric field, $E_\parallel = -ik_\parallel\phi -\pderiv{A_\parallel}{t}$:
\begin{equation}
    k_\perp^2\pderiv{A_\parallel}{t} = - \mu_0 e(i k_\parallel) \sum_j \frac{\Delta v_\parallel}{\sqrt{2}}\bar{v}_\parallel^{j\,2} f_0^j - \frac{\mu_0 e^2 n_0}{m_e}\left(ik_\parallel \phi +\pderiv{A_\parallel}{t}\right)C_N.
\end{equation}
This is essential for avoiding the cancellation problem because if we instead had $C_N\neq C_J$, we would have had a leftover term proportional to $(C_N-C_J)\pderiv{A_\parallel}{t}$ on the left-hand side. This leftover term would then lead to the spurious term  proportional to $\hat{\beta}/(k_\perp^2\rho_s^2)$ in Eq. (\ref{cancel-disp}).

In order to compute the integral quantities in the field equations, we use Eq. (\ref{f0w}) to compute
\begin{align}
    f_0^j &= -\frac{e}{m_e}\left( ik_\parallel \phi + \pderiv{A_\parallel}{t}\right)\frac{i}{\omega - k_\parallel \bar{v}_\parallel^j}\frac{\sqrt{2}}{\Delta v_\parallel} \widehat{F}_{Mh}^{j}\bigg\rvert_{-1}^{\ 1} \notag\\
    &\approx -\frac{e}{m_e}\left( ik_\parallel \phi + \pderiv{A_\parallel}{t}\right)\frac{i}{\omega}\left(1+\frac{k_\parallel \bar{v}_\parallel^j}{\omega} + \frac{k_\parallel^2 \bar{v}_\parallel^{j\,2}}{\omega^2} + \frac{k_\parallel^3 \bar{v}_\parallel^{j\,3}}{\omega^3}\right) \frac{\sqrt{2}}{\Delta v_\parallel} \widehat{F}_{Mh}^{j}\bigg\rvert_{-1}^{\ 1} \ \quad (\omega \gg k_\parallel v_{te}),
\end{align}
where we have expanded in the limit $\omega \gg k_\parallel v_{te}$.
Now we can calculate
\begin{gather}
    \sum_j \frac{\Delta v_\parallel}{\sqrt{2}} f_0^j = -\frac{e}{m_e}\left( ik_\parallel \phi + \pderiv{A_\parallel}{t}\right)\frac{i}{\omega}\sum_j\left(1+\frac{k_\parallel \bar{v}_\parallel^j}{\omega} + \frac{k_\parallel^2 \bar{v}_\parallel^{j\,2}}{\omega^2} + \frac{k_\parallel^3 \bar{v}_\parallel^{j\,3}}{\omega^3}\right)  \widehat{F}_{Mh}^{j}\bigg\rvert_{-1}^{\ 1} \\
    \sum_j \frac{\Delta v_\parallel}{\sqrt{2}}\bar{v}_\parallel^{j\,2}f_0^j = -\frac{e}{m_e}\left( ik_\parallel \phi + \pderiv{A_\parallel}{t}\right)\frac{i}{\omega}\sum_j\left(1+\frac{k_\parallel \bar{v}_\parallel^j}{\omega}\right) \bar{v}_\parallel^{j\,2} \widehat{F}_{Mh}^{j}\bigg\rvert_{-1}^{\ 1}
\end{gather}
Substituting these integral quantities into the field equations, the Poisson equation becomes
\begin{align}
    k_\perp^2 \frac{m_i n_0}{B^2}\phi &= \frac{e^2}{m_e}\left(ik_\parallel \phi + \pderiv{A_\parallel}{t}\right)\frac{i k_\parallel}{\omega^2}\sum_j \left[\bar{v}_\parallel^j \widehat{F}_{Mh}^{j}\bigg\rvert_{-1}^{\ 1} +\frac{k_\parallel}{\omega} \left(1+\frac{k_\parallel \bar{v}_\parallel^j}{\omega} \right)\bar{v}_\parallel^{j\,2} \widehat{F}_{Mh}^{j}\bigg\rvert_{-1}^{\ 1}\right] \notag \\
    &= \frac{e^2}{m_e}\left(ik_\parallel \phi + \pderiv{A_\parallel}{t}\right)\frac{i k_\parallel}{\omega^2} \left[-n_0 C_N +\sum_j\frac{k_\parallel}{\omega} \left(1+\frac{k_\parallel \bar{v}_\parallel^j}{\omega} \right)\bar{v}_\parallel^{j\,2} \widehat{F}_{Mh}^{j}\bigg\rvert_{-1}^{\ 1}\right] \label{poissonw}
\end{align}
and Ohm's law becomes
\begin{align}
    k_\perp^2 \pderiv{A_\parallel}{t} &= \frac{\mu_0 e^2}{m_e} \left(ik_\parallel \phi + \pderiv{A_\parallel}{t}\right)\sum_j \left[\bar{v}_\parallel^j \widehat{F}_{Mh}^{j}\bigg\rvert_{-1}^{\ 1} +\frac{k_\parallel}{\omega} \left(1+\frac{k_\parallel \bar{v}_\parallel^j}{\omega} \right)\bar{v}_\parallel^{j\,2} \widehat{F}_{Mh}^{j}\bigg\rvert_{-1}^{\ 1}\right] \notag \\
    &= \frac{\mu_0 e^2}{m_e} \left(ik_\parallel \phi + \pderiv{A_\parallel}{t}\right) \left[-n_0 C_N +\sum_j\frac{k_\parallel}{\omega} \left(1+\frac{k_\parallel \bar{v}_\parallel^j}{\omega} \right)\bar{v}_\parallel^{j\,2} \widehat{F}_{Mh}^{j}\bigg\rvert_{-1}^{\ 1}\right], \label{ohmw3}
\end{align}
where we have substituted the definition of $C_N$.
We can now combine Eqs. (\ref{poissonw}) and (\ref{ohmw3}) by multiplying Eq. (\ref{poissonw}) by $i k_\parallel T_e/n_0$, multiplying Eq. (\ref{ohmw3}) by $\rho_s^2 = m_i T_e/(e^2 B^2)$ and summing the two equations to get
\begin{align}
&k_\perp^2\rho_s^2\left(ik_\parallel \phi + \pderiv{A_\parallel}{t}\right)= \notag \\
&\quad \left(ik_\parallel \phi + \pderiv{A_\parallel}{t}\right)\left(\hat{\beta}-\frac{k_\parallel^2 v_{te}^2}{\omega^2}\right) \left[-C_N +\frac{1}{n_0}\sum_j\frac{k_\parallel}{\omega} \left(1+\frac{k_\parallel \bar{v}_\parallel^j}{\omega} \right)\bar{v}_\parallel^{j\,2} \widehat{F}_{Mh}^{j}\bigg\rvert_{-1}^{\ 1}\right],
\end{align}
with $\hat{\beta}=(\beta_e/2)m_i/m_e$. This then yields the dispersion relation
\begin{align}
k_\perp^2\rho_s^2 &=\left(\hat{\beta}-\frac{k_\parallel^2 v_{te}^2}{\omega^2}\right)\left[-C_N +\frac{1}{n_0}\sum_j\frac{k_\parallel}{\omega} \left(1+\frac{k_\parallel \bar{v}_\parallel^j}{\omega} \right)\bar{v}_\parallel^{j\,2} \widehat{F}_{Mh}^{j}\bigg\rvert_{-1}^{\ 1}\right]. \label{dispw}
\end{align}

To evaluate $C_N$ and the other sum, we need to project the background onto the basis in each cell. Taking $F_M = n_{0h} (2\pi v_t^2)^{-1/2}\exp\left(-v_\parallel^2/(2v_t^2)\right)$, we project onto the basis in cell $j$ as
\begin{align}
    F_{M0}^j &= \frac{1}{\sqrt{2}}\int\limits_{-1}^1 d\xi\frac{1}{\sqrt{2\pi v_t^2}}e^{\frac{-\left(\bar{v}_\parallel^j+\Delta v_\parallel \xi/2\right)^2}{2 v_t^2}} \notag \\
    &= \frac{n_{0h}}{\Delta v_\parallel \sqrt{2}}\left[ \text{erf}\left(\frac{(j+1/2)\Delta v_\parallel}{v_t\sqrt{2}}\right) - \text{erf}\left(\frac{(j-1/2)\Delta v_\parallel}{v_t\sqrt{2}}\right)\right] \\
    F_{M1}^j &= \frac{\sqrt{3}}{\sqrt{2}}\int\limits_{-1}^1 d\xi\ \xi\frac{1}{\sqrt{2\pi v_t^2}}e^{\frac{-\left(\bar{v}_\parallel^j+\Delta v_\parallel \xi/2\right)^2}{2 v_t^2}} \notag \\ &= -\frac{2 n_{0h}v_t}{\Delta v_\parallel^2}\sqrt{\frac{3}{\pi}}\left(e^{\frac{-\left((j+1/2)\Delta v_\parallel\right)^2}{2 v_t^2}} - e^{\frac{-\left((j-1/2)\Delta v_\parallel\right)^2}{2 v_t^2}} \right) \notag \\
    &\qquad-\frac{n_{0h}\sqrt{6}}{\Delta v_\parallel}j\left[ \text{erf}\left(\frac{(j+1/2)\Delta v_\parallel}{v_t\sqrt{2}}\right) - \text{erf}\left(\frac{(j-1/2)\Delta v_\parallel}{v_t\sqrt{2}}\right)\right],
\end{align}
where we have taken the cell centre to be $\bar{v}_\parallel^j = j\Delta v_\parallel$. Now we can evaluate integrated quantities such as
\begin{align}
    \sum_{j=-N}^P \frac{\Delta v_\parallel}{\sqrt{2}}F_{M0}^j &= \frac{n_{0h}}{2}\left[\text{erf}\left(\frac{(P+1/2)\Delta v_\parallel}{v_t\sqrt{2}}\right) - \text{erf}\left(\frac{-(N+1/2)\Delta v_\parallel}{v_t\sqrt{2}}\right)\right] \notag \\
    &= \frac{n_{0h}}{2}\left[\text{erf}\left(\frac{v_\text{max}}{v_t\sqrt{2}}\right) - \text{erf}\left(\frac{v_\text{min}}{v_t\sqrt{2}}\right)\right] \notag \\
    &= n_{0h}\text{erf}\left(\frac{v_\text{max}}{v_t\sqrt{2}}\right) \qquad\qquad \text{assuming } v_\text{min} = -v_\text{max}
\end{align}
where now note that we have finite limits on the sum to indicate finite extents of the $v_\parallel\in [-v_\text{max}, v_\text{max}]$ grid.
 As we alluded to before, we will define $n_{0h}$ so that $\sum_j \frac{\Delta v_\parallel}{\sqrt{2}}F_{M0}^j = n_0$ by definition, which means 
\begin{equation}
    n_{0h} = \frac{n_0}{\text{erf}\left(\frac{v_\text{max}}{v_t\sqrt{2}}\right)}.
\end{equation}
Note that $\text{erf}(x)$ quickly approaches 1 with increasing $x$, so that for example when $v_\text{max} = 4 v_t$, $n_{0h} \approx 1.00006 \,n_0$. We can also calculate
\begin{align}
    C_N &= -\frac{1}{n_0}\sum_{j=-N}^P \bar{v}_\parallel^j \widehat{F}_{M h}^j\bigg\rvert_{-1}^{\ 1} \notag \\
    &= -\frac{1}{n_0}\sum_{j=-N}^P \frac{j\Delta v_\parallel}{2}\left[ F_{Mh}^j(1)+F_{Mh}^{j+1}(-1) - F_{Mh}^j(-1)-F_{Mh}^{j-1}(1)\right] \notag \\&\qquad\qquad +\frac{\sigma}{n_0}  \sum_{j=-N}^P\frac{j\Delta v_\parallel}{2}\left[ F_{Mh}^{j+1}(-1)-F_{Mh}^{j}(1) - F_{Mh}^j(-1)+F_{Mh}^{j-1}(1) \right] \notag \\
    &= \frac{1}{n_0}\sum_{j=-N}^P \frac{\Delta v_\parallel}{2}\left[F_{Mh}^j(1)+F_{Mh}^j(-1)\right] \notag \\
    &\qquad\qquad+\frac{\sigma}{n_0} \sum_{j=-N}^P \frac{\Delta v_\parallel}{2}\left[F_{Mh}^j(1)-F_{Mh}^j(-1)\right]+ \text{boundary terms}\notag \\
    &=  \frac{1}{n_0}\sum_{j=-N}^P \frac{\Delta v_\parallel}{\sqrt{2}} F_{M0}^j +\frac{\sigma}{n_0} \sum_{j=-N}^P \frac{\Delta v_\parallel\sqrt{3}}{\sqrt{2}} F_{M1}^j +  \text{boundary terms} \notag \\
    &= 1 + \text{boundary terms} \approx 1,
\end{align}
where $\sigma$ is the sign of the upwind velocity, and the boundary terms that result from the finite limits on the sum are small for $v_\text{max}\gtrsim 4 v_t$. Thus we have $C_N\approx1$ as expected, although it does not need to be exactly equal to unity to eliminate the cancellation problem. Instead, it was sufficient that $C_N=C_J$ on either side of Eq. (\ref{ohmw2}).

One can also show that 
\begin{gather}
    \sum_{j=-N}^P \bar{v}_\parallel^{j\,2} \widehat{F}_{M h}^j\bigg\rvert_{-1}^{\ 1} = \text{boundary terms }\approx 0 \\
    \sum_{j=-N}^P \bar{v}_\parallel^{j\,3} \widehat{F}_{M h}^j\bigg\rvert_{-1}^{\ 1} =-3n_0 v_t^2\left(1 - \frac{\Delta v_\parallel^2}{12 v_t^2}\right) + \text{boundary terms }\approx - 3n_0 v_t^2\left(1 - \frac{\Delta v_\parallel^2}{12 v_t^2}\right).
\end{gather}
Now substituting these results into the dispersion relation from Eq. (\ref{dispw}), we obtain
\begin{align}
k_\perp^2\rho_s^2 \approx \left(\hat{\beta}-\frac{k_\parallel^2 v_{te}^2}{\omega^2}\right)\left(-C_N - \frac{3 k_\parallel^2 v_{te}^2}{\omega^2}\left(1 - \frac{\Delta v_\parallel^2}{12 v_{te}^2}\right)\right) \approx -\left(\hat{\beta}-\frac{k_\parallel^2 v_{te}^2}{\omega^2}\right)
\end{align}
after again taking the limit $\omega \gg k_\parallel v_{te}$ and assuming $\Delta v_\parallel \sim v_{te}$. This finally gives
\begin{equation}
    \omega^2 \approx \frac{k_\parallel^2 v_{te}^2}{\hat{\beta}+ k_\perp^2 \rho_s^2},
\end{equation}
which is the expected dispersion relation.

\bibliographystyle{jpp} 
\bibliography{nrmbib}

\end{document}